\documentclass[12pt]{article}
\usepackage[utf8]{inputenc}
\usepackage{fullpage}
\usepackage{authblk}

\usepackage{hyperref}

\usepackage{verbatim}
\usepackage[T1]{fontenc}
\usepackage[tt=false,type1=true]{libertine}
\usepackage[varqu]{zi4}
\usepackage[libertine]{newtxmath}

\usepackage{amsmath}
\usepackage{mathtools}
\usepackage{bm}
\usepackage{bbm}

\usepackage{amssymb}

\usepackage{amsthm}

\usepackage{thmtools}
\usepackage{thm-restate}

\usepackage{authblk}

\usepackage{pgfplots}
\pgfplotsset{compat=1.18}

\renewcommand{\vec}[1]{\bm{#1}}
\newcommand{\mat}[1]{\mathbf{#1}}

\newcommand{\R}{\mathbb{R}}

\newcommand{\N}{\mathbb{N}}

\newcommand{\interior}{\mathrm{int}}
\newcommand{\card}{\mathrm{card}}
\newcommand{\vx}{\vec{x}}

\newcommand{\vm}{\vec{m}}

\DeclareMathOperator*{\argmax}{arg\,max}

\newcommand{\defeq}{\coloneqq}

\usepackage{cleveref}

\let\R\Relax
\usepackage{complexity}
\let\R\relax
\newcommand{\R}{\mathbb{R}}
\usepackage{booktabs}

\usepackage{tikz}
\usepackage[american]{circuitikz}
\usetikzlibrary{patterns,shapes.geometric,arrows.meta, positioning, fit, calc, backgrounds}

\theoremstyle{plain}
\newtheorem{theorem}{Theorem}[section]

\newtheorem{lemma}[theorem]{Lemma}
\newtheorem{claim}[theorem]{Claim}

\theoremstyle{definition}
\newtheorem{definition}[theorem]{Definition}
\newtheorem{assumption}[theorem]{Assumption}

\theoremstyle{remark}
\newtheorem{remark}[theorem]{Remark}

\newcommand{\declarecolor}[2]{\definecolor{#1}{RGB}{#2}\expandafter\newcommand\csname #1\endcsname[1]{\textcolor{#1}{##1}}}
\declarecolor{White}{255, 255, 255}
\declarecolor{Black}{0, 0, 0}
\declarecolor{Maroon}{128, 0, 0}
\declarecolor{Coral}{255, 127, 80}
\declarecolor{Red}{182, 21, 21}
\declarecolor{LimeGreen}{50, 205, 50}
\declarecolor{DarkGreen}{0, 80, 0}
\declarecolor{Purple}{146, 42, 158}
\declarecolor{Navy}{0, 0, 128}
\declarecolor{LightBlue}{84, 101, 202}
\definecolor{mydarkblue}{rgb}{0,0.08,0.45}

\hypersetup{ %
    pdftitle={},
    pdfkeywords={},
    pdfborder=0 0 0,
    pdfpagemode=UseNone,
    colorlinks=true,
    linkcolor=Maroon,
    citecolor=Navy,
    filecolor=Purple,
    urlcolor=Purple,
}

\usepackage{natbib}
\title{Chaos in Autobidding Auctions}

\author[1]{Ioannis Anagnostides\thanks{Part of this work was performed while at Google DeepMind.}}
\author[2]{Ian Gemp}
\author[2,3]{Georgios Piliouras}
\author[4]{Kelly Spendlove}

\affil[1]{Carnegie Mellon University}
\affil[2]{Google DeepMind}
\affil[3]{Singapore University of Technology and Design}
\affil[4]{Google}
\affil[ ]{\small \texttt{ianagnos@cs.cmu.edu}, \texttt{imgemp@google.com}, \texttt{gpil@google.com},
\texttt{spendlove@google.com}}

\date{}

\begin{document}

\maketitle
\thispagestyle{empty}

\begin{abstract}
    As autobidding systems increasingly dominate online advertising auctions, characterizing their long-term dynamical behavior is brought to the fore. In this paper, we examine the dynamics of autobidders who optimize value subject to a return-on-spend (RoS) constraint under uniform bid scaling. Our main set of results show that simple autobidding dynamics can exhibit formally \emph{chaotic behavior}. This significantly strengthens the recent results of Leme, Piliouras, Schneider, Spendlove, and Zuo (EC '24) that went as far as quasiperiodicity.
    
    Our proof proceeds by establishing that autobidding dynamics can simulate---up to an arbitrarily small error---a broad class of continuous-time nonlinear dynamical systems. This class contains as a special case \emph{Chua's circuit}, a classic chaotic system renowned for its iconic double scroll attractor. Our reduction develops several modular gadgets, which we anticipate will find other applications going forward. Moreover, in discrete time, we show that different incarnations of mirror descent can exhibit Li-Yorke chaos, topological transitivity, and sensitivity to initial conditions, connecting along the way those dynamics to classic dynamical systems such as the logistic map and the Ricker population model.

    Taken together, our results reveal that the long-term behavior of ostensibly simple second-price autobidding auctions can be inherently unpredictable and complex.
\end{abstract}

\tableofcontents
\thispagestyle{empty}
\clearpage
\setcounter{page}{1}

\section{Introduction}
\label{sec:intro}

Online advertising has undergone a paradigm shift. While the classic theory of advertising auctions is predicated on the assumption that advertisers bid so as to maximize their quasi-linear utility---the difference between the obtained value and the payment, the modern ecosystem is increasingly being dominated by \emph{autobidding}~\citep{Aggarwal24:Autobidding}. Instead of manually submitting bids, which can be prohibitively complex in modern real-world settings, the advertiser specifies some high-level constraints to an automated agent, the \emph{autobidder}, who then optimizes on behalf of the advertiser across multiple auctions so as to optimally meet those objectives. One common constraint is \emph{return-on-spend (RoS)}, which demands a certain level of value per unit of dollars spent.

The prevalence of autobidding systems and value maximization subject to budget constraints can be attributed to multiple factors. First, it often aligns more closely with how firms design and execute advertising campaigns: advertisers oftentimes fix a budget \emph{a priori} and then seek to maximize coverage or impact subject to that budget constraint, rather than explicitly trading off value and payments, as the quasi-linear model would predict. It is also easier for advertisers to verify \emph{ex post} whether automated bidding platforms are achieving their specified goals, since performance can be judged based on delivered outcomes without relying on inferred, counterfactual utilities. At a more fundamental level, assigning a precise monetary value to ad clicks or impressions---a key premise in the quasi-linear model---is often impossible; the true economic value of advertising is difficult to quantify beyond crude proxy measures. Finally, (bidder) welfare maximization---as accomplished by the celebrated VCG mechanism---is typically at odds with optimizing (auctioneer) revenue~\citep{Ausubel06:Lovely,Myerson81:Optimal}.

The widespread deployment of autobidding agents drastically alters the market's behavior. A pressing question is to understand the dynamic behavior of systems in which multiple autobidders compete against each other. Do they converge to an equilibrium? If not, what type of behavior does the system exhibit in the long run? Previous work recently set out to investigate these fundamental questions. In particular, the paper of~\citet{Leme24:Complex}, the main precursor to our work, found that even simple autobidding dynamics can fail to stabilize, exhibiting recurrent, quasiperiodic behavior. However, the full extent of this instability was not understood. Quasiperiodicity is a benign and highly structured behavior that allows a considerable degree of predictability. The main open question is whether these systems can exhibit more erratic and complex behavior, rendering market outcomes inherently unpredictable.

\subsection{Our results: the onset of chaos}

\begin{figure}[t]
    \centering
    \includegraphics[scale=0.65]{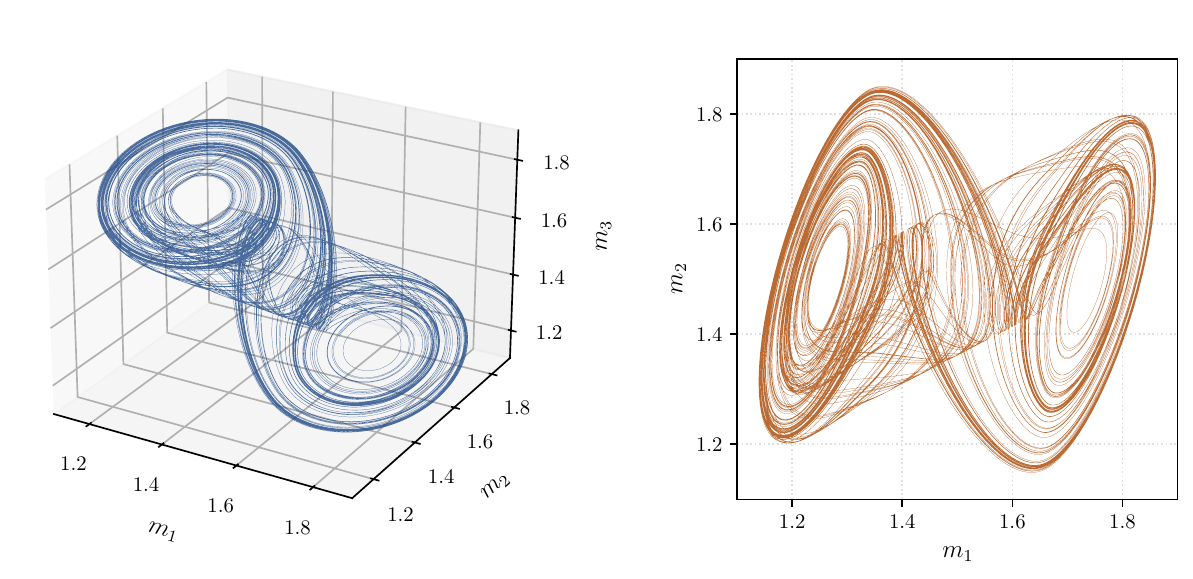}    
    \caption{Trajectories of autobidding dynamics.  }
    \label{fig:autobidding-chaos}
\end{figure}

We establish that autobidding dynamics are far more complex than previously understood, significantly strengthening the findings of~\citet{Leme24:Complex}. Specifically, we prove that they can exhibit formally \emph{chaotic behavior}. This means that the dynamics are beset by, among others, acute sensitivity to initial conditions and topological mixing. The practical ramifications are tangible: a tiny error in numerical simulation---for example, due to floating-point arithmetic or measurement noise---amplifies exponentially over time, rendering long-term prediction virtually impossible.

Our results cover both the discrete- and continuous-time settings. They also persist even in the simplified setting of pure RoS constraints (formally introduced in~\Cref{sec:prels}), absent any additional budget constraints. Since our goal is to establish complexity results, restricting to an ostensibly simple class of problems makes our results more surprising.

\subsubsection{Continuous-time dynamics} 

We analyze the continuous-time dynamical system put forward by~\citet{Leme24:Complex}, arguably the most natural and simple algorithm for optimizing under RoS constraints; it is also closely related to other recently analyzed dynamics with strong theoretical properties~\citep{Gaitonde23:Budget,Lucier24:Autobidders}, as discussed in more detail in \Cref{sec:related}. To place our approach in context, we highlight that establishing chaos for continuous-time systems is notoriously challenging. For example, characterizing the attractor of Lorenz's iconic system---one of the most well-studied dynamical systems---took many decades to be finally resolved via a computer-assisted proof, and was the subject of Stephen Smale's 14th problem~\citep{Tucker02:Rigorous}. 

In this context, instead of characterizing autobidding systems from scratch, we make use of reductions. In particular, we prove that a broad class of nonlinear dynamical systems can be simulated---up to an arbitrarily small error in the vector field---by autobidding dynamics under the appropriate market competition (\Cref{theorem:simulation-general}). This class contains as a special case \emph{Chua's circuit}~\citep{Matsumoto03:Chaotic}, a classic chaotic system. \Cref{fig:autobidding-chaos} illustrates the autobidding dynamics that result from our reduction, which (approximately) simulate Chua's circuit. The main takeaway is that long-term forecasting of market states can be intractable, and small perturbations can lead to macroscopically different outcomes.

\paragraph{Proof of our general simulation result}

From a technical standpoint, our reduction develops a series of gadgets---configurations of items, values, and reserve prices---to implement specific dynamical system functionalities. The conceptual structure of our reduction is illustrated in~\Cref{fig:reduction}.

Our simulation result encompasses any (bounded) nonlinear system in which each nonlinearity depends solely on the corresponding state variable. In other words, we exclude \emph{coupled} nonlinear interactions such as the cross terms $x_i x_{i'}$ present in, for example, the Lorenz attractor. Chua's circuit (\Cref{sec:chua}) is a canonical member of this class, as its nonlinearity---typically a piecewise-linear function---depends only on the voltage across a single component.

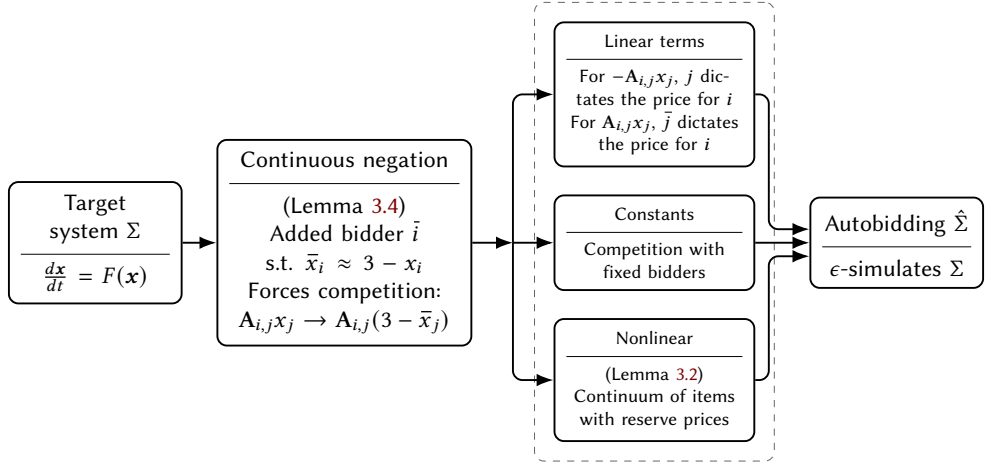
\begin{figure}[t]
    \centering
    \scalebox{0.9}{

\tikzset{
    base/.style={
        rectangle, rounded corners, draw=black, line width=0.8pt,
        align=center, font=\sffamily\footnotesize,
        inner sep=5pt
    },
    startend/.style={
        base, 
        text width=2.2cm, 
        minimum height=1cm
    },
    transform/.style={
        base, 
        text width=3.4cm, 
        minimum height=1.5cm
    },
    branch/.style={
        base, 
        text width=2.6cm, 
        font=\sffamily\scriptsize
    },
    arrow/.style={
        -{Latex[length=2.5mm]}, line width=0.8pt, draw=black
    },
    group_label/.style={
        font=\sffamily\scriptsize\itshape, 
        color=black!70
    }
}

\begin{tikzpicture}[node distance=0.4cm and 0.5cm]

    \node[startend] (start) {
        Target system $\Sigma$
        \\[-0.5em] \rule{\linewidth}{0.4pt} \\ 
        $\frac{d \vx}{dt} = F(\vx)$
    };

    \node[transform, right=of start] (negation) {
        Continuous negation
        \\[-0.5em] \rule{\linewidth}{0.4pt} \\ 
        (\Cref{lemma:continuousnegation})\\
        Added bidder $\widebar{i}$ s.t.
        $\widebar{x}_i \approx 3-x_i$ \\
        Forces competition:\\
        $\mat{A}_{i, j} x_j \to \mat{A}_{i, j}(3 - \widebar{x}_j)$
    };

    
    \node[branch, right=1.2cm of negation] (const) {
        Constants 
        \\[-0.5em] \rule{\linewidth}{0.4pt} \\ 
        Competition with fixed bidders
    };
    
    \node[branch, above=of const] (linear) {
        Linear terms 
        \\[-0.5em] \rule{\linewidth}{0.4pt} \\ 
        For $- \mat{A}_{i, j} x_j$, $j$ dictates the price for $i$\\
        For $\mat{A}_{i, j} x_j$, $\widebar{j}$ dictates the price for $i$
    };
    
    \node[branch, below=of const] (nonlinear) {
        Nonlinear
        \\[-0.5em] \rule{\linewidth}{0.4pt} \\ 
        (\Cref{lemma:nonlinear-sim})\\
        Continuum of items with reserve prices
    };

    \node[startend, right=0.8cm of const] (final) {
        Autobidding $\hat{\Sigma}$
        \\[-0.5em] \rule{\linewidth}{0.4pt} \\ 
        $\epsilon$-simulates $\Sigma$
    };

    
    \draw[arrow] (start) -- (negation);
    
    \coordinate (split) at ($(negation.east) + (0.6cm,0)$);
    \draw[arrow] (negation.east) -- (split);

    \draw[arrow, rounded corners] (split) |- (linear.west);
    \draw[arrow] (split) -- (const.west);
    \draw[arrow, rounded corners] (split) |- (nonlinear.west);

    \draw[arrow, rounded corners] (linear.east) -- ++(0.2cm,0) |- ($(final.west) + (0, 0.2cm)$);
    \draw[arrow] (const.east) -- (final.west);
    \draw[arrow, rounded corners] (nonlinear.east) -- ++(0.2cm,0) |- ($(final.west) + (0, -0.2cm)$);

    \begin{pgfonlayer}{background}
        \node[fit=(linear)(nonlinear)(const), rounded corners, draw=black!60, dashed, inner sep=8pt] (groupbox) {};
    \end{pgfonlayer}

\end{tikzpicture}}
    \caption{The key steps in our main simulation result (\Cref{theorem:simulation-general})}
    \label{fig:reduction}
\end{figure}

The first essential gadget is what we refer to as \emph{continuous negation}. To appreciate its importance, we recall that autobidding systems are, in a certain regime, \emph{competitive} systems~\citep{Leme24:Complex}. This arises because the price paid by an autobidder is determined by the (bid) multiplier of a competitor, imposing specific constraints on the interaction structure. \citet{Leme24:Complex} observed that for \emph{linear} dynamical systems, this constraint is not restrictive in that one can always apply a linear transformation to map a general linear system to a competitive one. However, in the nonlinear setting, such coordinate transformations are problematic because they distort the specific structure of the nonlinearities we aim to simulate.

Our continuous negation gadget addresses this challenge by converting the linear part of any nonlinear system into a competitive form \emph{without} coordinate changes (\Cref{fig:reduction}). Specifically, it takes as input the multiplier of an input autobidder $m(t)$, and produces an output multiplier that satisfies $\widebar{m}(t) \approx 3 - m(t)$ (the constant $3$ has no special significance). This allows us to introduce an auxiliary autobbider $\widebar{i}$ so that non-competitive (positive feedback) interactions with $i$ are replaced by competitive (negative feedback) interactions with $\widebar{i}$.

We construct this gadget by reverse engineering an autobidding system governed by the differential equation 
\[
\frac{d \widebar{m} }{dt} = 3 \lambda - \lambda m(t) - \lambda \widebar{m}(t),
\]
parameterized by $\lambda > 0$. We prove that when $\lambda = \lambda(\epsilon)$ is large enough, $\widebar{m}(t)$ approaches $3 - m(t)$ to any desired precision $\epsilon > 0$ uniformly in time (\Cref{lemma:approxint}). This appoximation is the only source of error in our simulation.

Having converted the linear component into a competitive form, our next gadget constructs the nonlinear term. A key observation is that, after a suitable preprocessing, the target nonlinearity can be assumed to have a negative derivative. This monotonicity is crucial: it allows us to design a continuum of items and reserve prices that exactly recovers the underlying nonlinearity (\Cref{lemma:nonlinear-sim}). Specifically, the density of the items is dictated by the derivative of the nonlinear function. This step can also be (approximately) implemented with a discrete set of items by a standard integral approximation (\Cref{lemma:integral-approx}).

We couple these modules with the gadgets developed by~\citet{Leme24:Complex} to arrive at our main simulation result (\Cref{theorem:simulation-general}). A crucial aspect of our result is that chaos is not merely an artifact of discretization or specific learning rate choices, as is often the case in prior work, but rather an intrinsic property of the continuous-time dynamics.

We then analyze the system resulting from the simulation of Chua's circuit. First, we numerically verify the existence of a positive Lyapunov exponent (\Cref{sec:posLLE}), the hallmark of chaotic sensitivity whereby initially close trajectories diverge exponentially. Second, we address the structural robustness of chaos under the perturbations introduced by our continuous negation gadget. Leveraging the framework of~\citet{Galias97:Positive}---who provided a computer-assisted proof of positive topological entropy in Chua's circuit---we find that for $\lambda \gg 1$ (in which case the continuous negation error is negligible), the key topological preconditions of~\citet{Galias97:Positive} pertaining to the (deformed) horseshoe map are preserved (\Cref{sec:pos-topentr}). This confirms that the presence of chaos is structurally robust under the approximations inherent in our reduction.

\subsubsection{Discrete-time dynamics} 

In the second part of the paper, we examine discrete-time approximations of the previously studied continuous-time dynamics. Specifically, we analyze different incarnations of the classic \emph{mirror descent} algorithm under standard regularizers---namely, the entropic and the Euclidean one.

We show that even in a minimal setting with two autobidders and two items, the dynamics can exhibit \emph{Li-Yorke chaos} (\Cref{theorem:Li-Yorke-auto}). Formally, this means that, for certain learning rates, the trajectory of multipliers admits a dense set of periodic points and uncountable \emph{scrambled sets} (\Cref{theorem:LiYorke}). The underlying instance is simple, and has been studied before as it exposes the inefficiency of the VCG auction under value maximization with RoS constraints.

Furthermore, we also establish topological transitivity and sensitivity to initial conditions, satisfying the standard definition of Devaney chaos (\Cref{thm:Devaney}). We do so by connecting autobidding dynamics to the \emph{Ricker population model} (\Cref{thm:recover-Ricker}) and the classic \emph{logistic map} (\Cref{theorem:logistic}). Unlike our continuous-time results, which required a delicate construction to simulate Chua's circuit, the discrete-time instances we construct are minimal and based on natural conditions likely to arise in practice. On the other hand, our discrete-time results hinge on having a large enough learning rate, so in that sense they are less robust relative to the continuous-time counterparts.

Viewed as a whole, our paper brings together decades of dynamical systems research with the emerging study of autobidding systems.
\subsection{Related work}
\label{sec:related}

Our work builds upon several key strands of literature, which we discuss below.

\paragraph{Autobidding systems} There is a growing line of work that examines various aspects of autobidding systems. For an excellent survey covering recent advances, we refer to~\citet{Aggarwal24:Autobidding}. Closer to our paper are the works of~\citet{Gaitonde23:Budget} and~\citet{Lucier24:Autobidders}, which show that the liquid welfare attained when all autobidders adopt certain gradient-based dynamics is at least half the optimal one; these results does not hinge on the convergence of the dynamics, thereby strengthening the price of anarchy bound of~\citet{Aggarwal19:Autobidding} that only applies with respect to fixed points---also known as~\emph{autobidding equilibrium}~\citep{Li24:Vulnerabilities}. The dynamics considered in those papers are closely related to the ones we study in this paper. From a broader standpoint, characterizing the (liquid) welfare of autobidding markets has received extensive attention in recent years (\emph{e.g.},~\citealp{Deng24:Efficiency,Balseiro21:Robust,Fikioris23:Liquid,Baldeschi26:Optimal}). Our focus in this paper is different, examining whether the trajectories of autobidding dynamics settle in a predictable pattern.

\paragraph{Chua's circuit} Chua's circuit holds a distinct place in the history of dynamical systems, being the subject of intense study and extensively referenced in the literature. It is a canonical physical realization of chaotic dynamics, which can be easily implemented as an electronic circuit. A salient aspect of that system is its structural robustness. Unlike other systems whose chaotic behavior is brittle, Chua's circuit exhibits chaos over a broad range of parameters. The physical consequence of this is that even if the constituents of the system (resistors, inductors, and capacitors) are imperfect, chaos persists. By virtue of our reduction, autobidding dynamics inherit this robustness.

\paragraph{Chaos in multiagent settings} Chaotic behavior has been extensively documented in multiagent settings under various dynamics~\citep{Palaiopanos17:Multiplicative,Bielawski21:Follow,Cheung20:Chaos,Sato02:Chaos}. In particular, it is straightforward to simulate Lorenz's system through replicator dynamics. What distinguishes our results from that line of work is that autobidding dynamics have their own particular structure, so it was hitherto unclear whether they could sustain chaos, especially in continuous time. For example, it is worth pointing out that the class of nonlinear systems covered by our simulation result does not contain Lorenz's system; whether that can be (even approximately) simulated by autobidding dynamics is left as a challenging open question.

\paragraph{Competitve systems} As was observed by~\citet{Leme24:Complex}, autobidding dynamics behave, at least in a certain regime, as a competitive system; the basic reason for this is that bidders compete with each other for the items. Based on this insight, \citet{Leme24:Complex} went on to connect autobidding dynamics to \emph{repressilators} in synthetic
biology~\citep{Elowitz00:Synthetic}, a genetic regulatory network consisting of at least one negative feedback loop. Relatedly, \citet{Morrison24:Diversity} investigated competitive threshold-linear networks, demonstrating how the underlying graph connectivity dictates the system's emergent dynamics, which range from stable limit cycles and quasiperiodic attractors to chaos. Although the class of nonlinear systems considered by~\citet{Morrison24:Diversity} is distinct from the one covered by our simulation result (\Cref{theorem:simulation-general}), it would be interesting to understand whether the results of~\citet{Morrison24:Diversity} have any implications in the autobidding setting.

\paragraph{Complexity and dynamical systems} Finally, there is a significant body of work at the intersection of complexity theory and dynamical systems that establishes that many natural problems concerning the long-term system behavior are inherently intractable~\citep{Papadimitriou16:Computational,Chatziafratis19:Computational}. A natural question is whether autobidding dynamics, which are governed by specific budget or RoS constraints, are subject to similar complexity barriers.
\section{Preliminaries}
\label{sec:prels}

This section provides some basic preliminaries on value maximization with RoS constraints and autobidding dynamics. Additional background on discrete- and continuous-time dynamical systems is deferred to~\Cref{sec:moreprels,app:galias_framework}.

\subsection{Value maximization with RoS constraints}
\label{prels:ros}

Our main focus in this paper is on the behavior of autobidding dynamics, which capture the competition of multiple agents adaptively bidding for a set of items. Here, the advertiser, or buyer, provides as input to the autobidder a \emph{return-on-spend (RoS)} target, $\tau_i \geq 1$, which is assumed to remain fixed. Another common input in practice and prior literature alike is a \emph{budget constraint}~\citep{Dobzinski14:Efficiency}, which forces the total expenditure of each agent $i$ to remain below a given threshold $B_i > 0$; all our complexity results hold by lifting the budget constraint, which can be thought of as taking $B_i = + \infty$.

In more detail, we consider a setting comprising $n$ \emph{autobidders}. The auctioneer is to allocate $k \in \N$ items. The interaction between the auctioneer and the autobidders proceeds over a sequence of rounds. At every time $t$, each autobidder $i$ submits a bid $b_{i j}^{(t)}$ for item $j \in [k]$. We assume throughout that each autobidder $i$ has a perfect model of values, denoted by $\{ v_{ij} \}_{j \in [k]}$, unbeknownst to the auctioneer. Valuations are assumed to be \emph{additive}, so that the value $v_i(S)$ of a bundle $S \subseteq [k]$ is equal to $\sum_{j=1}^k v_{i j}$. The auctioneer takes as input the bids, whereupon it produces an \emph{allocation} $\vec{x} \in \R_{\geq 0}^{n \times k}$, with $\sum_{i=1}^n x_{ij} = 1$ for each item $j \in [k]$, together with \emph{payments} $\vec{p} \in \R^{n \times k}_{\geq 0}$. We consider the usual second-price allocation rule---which is an instantiation of the Vickrey-Clarke-Groves (VCG) mechanism~\citep{Vickrey61:Counterspeculation,Clarke71:Multipart,Groves73:Incentives}---with uniform tie breaking, whereby
\begin{equation}
    \label{eq:uni-tiebreaking}
    x_{ij}(\vec{b}) \defeq \frac{ \mathbbm{1} \{ b_{ij} = \max_{1 \leq i' \leq n} b_{i' j} \} }{ | \{ i' : b_{i' j} = \max_{1 \leq i'' \leq n} b_{i'' j} \} | } \quad 
\end{equation}
and
\[
p_{ij}(\vec{b}) \defeq x_{i j}(\vec{b}) \cdot \max_{i' \neq i} b_{i' j}.
\]

In words, the allocation of each item proceeds separately, reflecting the fact that valuations are additive. For each item $j \in [k]$, only bidders who submitted the highest bid obtain a positive fraction of the item and pay the second highest bid. A fractional allocation can be thought of in probabilistic terms by interpreting $x_{ij}$ as the probability that $i$ will receive item $j$. We posit uniform tie breaking for simplicity, in accordance with~\citet{Leme24:Complex}, although all our results are agnostic to that; it is worth pointing out that guaranteeing existence of autobidding equilibria requires using more careful tie breaking than~\eqref{eq:uni-tiebreaking}, but that is less of a concern in the dynamic setting~\citep{Li24:Vulnerabilities}.

\begin{remark}[Continuum of items]
    For mathematical convenience, we sometimes assume that there is a continuum of items to be allocated. Under mild assumptions on the density function, this can be approximated to any precision through the use of a sufficiently large number of items (\emph{cf.}~\Cref{lemma:integral-approx}).
\end{remark}

\begin{remark}[Reserve prices]
    Some of our constructions make use of \emph{reserve prices}, which are standard both in practice and in prior literature on the subject. If an item has a reserve price attached to it, it can be thought of as having an additional, fictitious autobidder submit that bid for the item. \citet{Li24:Vulnerabilities} have shown how to simulate the presence of reserve prices in equilibrium via auxiliary autobidders.
\end{remark}

It is well-known that VCG maximizes welfare in the \emph{quasi-linear} model, in which each agent $i$ is assumed to be maximizing its utility $\sum_{j=1}^k (v_{i j} x_{ij}(\vec{b}) - p_{i j} (\vec{b}))$. In fact, in that model, truthful reporting is a dominant strategy for each agent, so there is an obvious optimal bidding strategy \emph{regardless} of how other bidders engage; from a dynamical standpoint, this means that learning algorithms exhibit convergent behavior in that setting (\emph{e.g.},~\citealp{Bichler25:Beyond}). However, autobidding systems depart from the normative quasi-linear model in auction theory.

\paragraph{RoS objective} Instead of optimizing (quasi-linear) utility, the goal of each autobidder $i \in [n]$ is to submit a set of bids so as to maximize the allocation \emph{value} $\sum_{j = 1}^k x_{ij} v_{ij}$ subject to the \emph{return-on-spend (RoS)} constraint parameterized by $\tau_i \geq 1$; that is,
\begin{equation}
    \max_{\vec{b}_i} \sum_{j=1}^k v_{i j} x_{i j}(\vec{b}) \quad \text{subject to} \quad \tau_i \cdot \sum_{j=1}^k p_{i j}(\vec{b}) \leq \sum_{j=1}^k v_{i j} x_{i j}(\vec{b}).
\end{equation}

The prevalence of such constraints and the subsequent departure from the quasi-linear model in auction theory can be attributed to multiple factors, as we detailed in~\Cref{sec:intro}. By rescaling the values, it is without loss of generality to take $\tau_i = 1$ for each autobbider $i \in [n]$. In that case, the constraint takes the form of \emph{individual rationality}, which means that the expected utility from participating is nonnegative.

\paragraph{Uniform bid scaling} A natural strategy commonly employed in practice is for an autobidder $i \in [n]$ to select bids as $b_{i j} = m_i \cdot v_{i j}$, for some \emph{multiplier} $m_i > 0$. This is referred to as \emph{uniform bid scaling}. It has the effect of reducing the problem of the autobidder down to a single variable. For auction formats that are truthful under quasi-linear utilities, it has been shown that uniform bid scaling is optimal~\citep{Aggarwal19:Autobidding}; On the other hand, for auction formats that are not truthful---notably first-price auctions (FPAs) or generalized second-price auctions (GSPs)---uniform bid scaling can lead to suboptimal strategies; this was investigated extensively by~\citet{Deng21:Towards}.

Solving this optimization problem from the perspective of a single agent is straightforward: each agent should select the maximal multiplier that guarantees nonnegative utility (\Cref{fig:utility}). What makes autobidding interesting is that multiple agents are simultaneously striving to achieve that goal, resulting in each agent chasing a moving target. In fact, we establish that complex, chaotic behavior can emerge even with a single autobidder in the presence of reserve prices (\Cref{sec:Ricker}).

\begin{figure}
    \centering
    \scalebox{0.8}{\begin{tikzpicture}
    \begin{axis}[
        axis lines = center,
        xlabel = {$m$ (Multiplier)},
        ylabel = {Utility},
        ymin = -0.4, ymax = 1.3,
        xmin = 0, xmax = 3.0,
        xtick = {1},
        xticklabels = {$m=1$},
        ytick = \empty,
        axis line style = {-stealth},
        xlabel style={at={(current axis.right of origin)}, anchor=west, xshift=2pt},
        ylabel style={at={(axis description cs:0.08,1)}, anchor=south},
        clip=false
    ]

    \addplot [
        domain=0:1, 
        samples=50, 
        color=blue, 
        very thick,
    ]
    {1 - (x-1)^2};

    \addplot [
        color=blue,
        very thick,
        smooth,
        tension=0.5
    ] coordinates {
        (1, 1)      
        (1.5, 0.8)  
        (2.0, 0.4)  
        (2.4, 0)    
        (2.7, -0.3) 
    };

    \draw[dashed, gray] (axis cs:1,0) -- (axis cs:1,1);
    \node[circle, fill=black, inner sep=1.5pt] at (axis cs:1,1) {};
    
    \node[circle, fill=red, inner sep=1.5pt] at (axis cs:2.4,0) {};

    \end{axis}
\end{tikzpicture}}
    \caption{An illustrative utility plotted as a function of the multiplier. Under VCG, selecting a multiplier $m = 1$ can always guarantee the maximum utility, but that is typically not the optimal choice for value maximization under an RoS constraint. Increasing the multiplier beyond 1---which corresponds to overbidding---can decrease its quasi-linear utility since it can pay for an item more than its value for it. The goal is to select the largest possible multiplier that guarantees nonnegative utility; this corresponds to the red dot in the figure.}
    \label{fig:utility}
\end{figure}
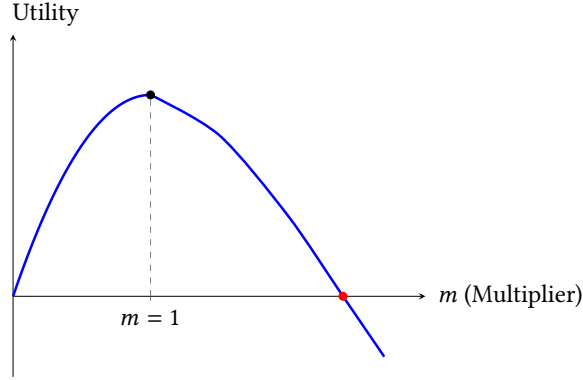

\subsection{Autobidding dynamics}

Autobidding dynamics emerge from the repeated interactions of autobidders with the auction mechanism, occurring even when a single autobidder adjusts bids to satisfy the RoS constraints. In what follows, we denote by $u_i(\vec{m}) \defeq \sum_{j=1}^k v_{i j} x_{i j}(\vec{m}) - \sum_{j=1}^k p_{i j}(\vec{m})$ the quasi-linear utility of autobidder $i$. Following~\citet{Leme24:Complex}, we study the following continuous-time dynamical system:
\begin{equation}
    \label{eq:continuous-time}
    \frac{ d m_i }{d t} = u_i(\vec{m}) \quad \forall i = 1, \dots, n.
\end{equation}

This is perhaps the most natural dynamical system to consider in this setting. Intuitively, $u_i(\vec{m})$ represents the slack in $i$'s constraint. When $u_i(\vec{m}) > 0$, the agent is overperforming relative to their target, prompting an increase in the multiplier in order to acquire more items. Conversely, when $u_i(\vec{m}) < 0$, the agent is violating the constraint, causing a reduction in the multiplier. If the multipliers remain bounded within $[0, M]$, it follows that~\eqref{eq:continuous-time} satisfies the RoS constraint on average up to an error of at most $M/T$~\citep{Leme24:Complex}. 

Our results in~\Cref{sec:discrete} concern discrete-time versions of the dynamical system given in~\eqref{eq:continuous-time}; for related dynamics, we refer to~\citet{Gaitonde23:Budget} and~\citet{Lucier24:Autobidders}.

\section{Simulating Chua's circuit and continuous-time autobidding chaos}
\label{sec:Chua}

This section shows that the continuous-time autobidding dynamics~\eqref{eq:continuous-time} can exhibit chaotic behavior, in a sense that will be formalized shortly. As we explained before, proving that from scratch is known to be especially challenging, so our basic approach is to show that autobidding dynamics can \emph{simulate} a known chaotic dynamical system---namely, Chua's circuit~\citep{Matsumoto03:Chaotic}. We are then able to leverage existing techniques developed for the analysis of Chua's circuit.

\subsection{Simulating a class of nonlinear systems}
\label{sec:sim-general}

In fact, our simulation result is more general, covering a broad class of nonlinear dynamical systems that contains Chua's circuit as a special case. Specifically, we consider the following class of continuous-time systems.
\begin{equation}
    \label{eq:nonlinear}
    \frac{d x_i}{dt} = \langle \mat{A}_{i, :}, \vec{x} \rangle + h_i(x_i) \quad \forall i = 1, 2, \dots d,
\end{equation}
where $\mat{A} \in \R^{d \times d} $ is a matrix, $\mat{A}_{i, :}$ denotes the $i$th row of $\mat{A}$, and each function $h_i(\cdot)$ is potentially nonlinear but continuous. The key structural assumption made above is that each nonlinear function \emph{depends only on the corresponding variable}. In other words, there are no coupling terms of the form $h(x_i, x_{i'})$ for $i \neq i'$. This excludes, for example, systems like Lorenz's which contains coupled quadratic terms; extending our simulation result beyond~\eqref{eq:nonlinear} is left as an interesting, but challenging direction for the future.

Our only additional (mild) assumption concerning~\eqref{eq:nonlinear} is that its orbits remain bounded, and this property is \emph{structurally robust} in that it is maintained when the vector field is subjected to a small perturbation. Interesting physical systems, like Chua's circuit, have such structural robustness, for otherwise their behavior would be brittle and fundamentally alter under marginal changes in the model parameters. Our assumptions are gathered below.

\begin{assumption}
    \label{assumption:nonlinear}
    There exists $B > 0$ such that for any time $t \geq 0$ and $i \in [d]$, $x_i(t) \in [-B, B]$. In particular, by rescaling and shifting, we can assume without any loss of generality that $x_i(t) \in [1.1, 1.9]$ for all $i \in [d]$. This is so when the underlying vector field $F$ is replaced by any $\hat{F}$ with $\|F(\vx) - \hat{F}(\vx) \| \leq \epsilon$ for any $\vx \in [-B, B]^d$, where $\epsilon > 0$ is sufficiently small.
    
    Furthermore, we assume that each function $h_i$ is continuous and differentiable almost everywhere, so that $|h_i'(\cdot)| \leq L$ for any $i \in [d]$.
\end{assumption}

The interval $[1.1, 1.9]$ is used for concreteness; there is nothing special about these numbers. By interpreting each variable as a multiplier, what is more important is maintaining the invariance that each multiplier always exceeds the value 1. We want to avoid establishing chaos through dynamics in which some autobidders sometimes use a multiplier below 1, because these are dominated strategies and are thus likely to be avoided by sophisticated autobidders.

We also point out that our subsequent construction can even handle the case where each function $h_i$ has a finite number of discontinuities, but working with non-continuous systems is more cumbersome, so we restrict our attention to the continuous setting.

We start by making a simple but crucial observation: by adding and subtracting the term $(L+1) x_i - (L+1) x_i$ in~\eqref{eq:nonlinear}, one obtains an equivalent system of the form
\begin{equation*}
    \frac{d x_i}{dt} = \langle \Tilde{\mat{A}}_{i, :}, \vec{x} \rangle + \Tilde{h}_i(x_i) \quad \forall i = 1, 2, \dots, d,
\end{equation*}
where $\Tilde{h}_i(x_i) = h_i(x_i) - (L+1) x_i$. This means that $\Tilde{h}_i'(\cdot) < 0$. In other words, we can assume without loss of generality that each nonlinear function $h_i$ in~\eqref{eq:nonlinear} has a negative derivative: $h'_i(\cdot) < 0$. This is important for the first part of our construction. To keep the notation as simple as possible, in the sequel, $\mat{A}$ plays the role of $\tilde{\mat{A}}$ and $h$ that of $\tilde{h}$.

\subsubsection{Nonlinear function}
\label{subsubsec:nonlinear}

We will first show how to create an instance in which the utility function of an autobidder contains any nonlinear term with negative derivative.

\begin{restatable}{lemma}{nonlinear}
    \label{lemma:nonlinear-sim}
Let $h: [1.1, 1.9] \to \R$ be any continuous and differentiable almost everywhere function such that $h'(\cdot) < 0$. There is an instance in which the utility of the autobidder is equal to $h$, modulo an additive constant. 
\end{restatable}

The presence of the additive constant is moot, because it can be pushed into the linear term. The construction is based on an instance with a continuum of items with reserve prices and a carefully constructed density function (\Cref{sec:proofsChua}). 

Furthermore, we observe that this construction is possible even with a finite number of items, with the caveat that one can only \emph{approximate}---uniformly, to any precision---the nonlinear function; our next gadget (\Cref{lemma:continuousnegation}) inherently incurs an approximation error, so some additional error in~\Cref{lemma:nonlinear-sim} does not affect our main simulation result. To do so, the first idea is to quantize the density and the reserve prices constructed in~\Cref{lemma:nonlinear-sim}---per the usual Riemann sums. This makes the density supported on a discrete number of prices, but its value may still be a non-integer, corresponding to allotting a fraction of each item. To recover such fractional allocations, one can use multiple \emph{symmetric} autobidders facing the same problem. Under the uniform tie breaking rule (\Cref{prels:ros}), this means that each autobidder obtains a suitable fraction of the item.

\begin{lemma}
    \label{lemma:integral-approx}
    Let $h: [1.1, 1.9] \to \R$ be any continuous and differentiable almost everywhere function such that $h'(\cdot) < 0$. For any precision $\epsilon > 0$, there is an instance with a sufficiently large number of items $k = k(\epsilon)$ and autobidders $n = n(\epsilon)$ such that
    \begin{itemize}
        \item if $m_1 = m_2 = \dots = m_n = m$, then $u_i(m)$ is uniformly $\epsilon$-close to $h$ for any $i \in [n]$, modulo some additive constant, and
        \item the system is symmetric, so that a symmetric initialization produces symmetric multipliers.
    \end{itemize}
\end{lemma}

\subsubsection{Continuous negation}

The next central piece in our construction is what we refer to as \emph{continuous negation}. To explain its role, we first need to point out that autobidding systems are, at least in a certain regime, \emph{competitive}~\citep{Leme24:Complex}. What this means is that when one autobidder increases its multiplier, this has a negative externality on the other autobidders' utilities on account of the concomitant increase in the prices. This is a considerable obstacle in simulating even linear systems, let alone nonlinear ones. In particular, while~\citet{Leme24:Complex} observed that one can apply a linear transformation to map a general linear system to a competitive one, in the nonlinear setting such coordinate transformations are prohibitive because they affect the specific structure of the system we aim to simulate---namely~\eqref{eq:nonlinear}---by introducing nonlinear coupling terms.

Continuous negation addresses this technical challenge. It takes as input a multiplier $m$, and, without affecting the evolution of the input, produces an output multiplier $\widebar{m} = 3 - m$, which we think of as the negation of $m$. The only caveat is that the implementation incurs some small approximation error, which can be made arbitrarily small.

\begin{restatable}{lemma}{contneg}
    \label{lemma:continuousnegation}
    Let $m(t) \in [1.1, 1.9]$ be the multiplier of some autobidder. There is an instance which does not affect the evolution of $m(t)$ and produces an autobidder whose multiplier $\widebar{m}(t)$ evolves as
    \begin{equation*}
        \widebar{m}(t) = e^{-\lambda t} \int_{\tau=0}^t \lambda e^{\lambda \tau} ( 3 - m(\tau)) d \tau + e^{- \lambda t} \widebar{m}(0) \quad \forall t \geq 0.
    \end{equation*}
\end{restatable}

\begin{proof}
    Let $\lambda \geq 1$ be some parameter. 
    We consider a new output autobidder. The input autobidder competes with the output bidder for a single item valued $\lambda$ and $1.95 \lambda$, respectively. By maintaining the invariance that $\widebar{m}(t) \approx 3 - m(t)$ (\Cref{lemma:approxint}), and in particular $\widebar{m}(t) > 1$, we can thus guarantee that the output autobidder always wins this item, so the utility function of the input autobidder remains unaffected. The corresponding payment is dictated by the multiplier of the input autobidder, and is equal to $\lambda m(t)$.
    
    Moreover, we consider a continuum of items, which we represent with a density function $\rho : [1.05, 2] \to \R_{\geq 0}$. We parameterize $\rho$ in terms of the reserve price attached to each item in the continuum. The output autobidder values these items equally for 1 per unit. (Only the output autobidder values these items.) The density is set to $\rho(p) = \lambda/(p-1)$ for $p \geq 1.05$. As a result, the utility of the output autobidder reads
    \begin{equation*}
        1.95 \lambda - \lambda m(t) + \int_{1.05}^{\widebar{m}(t)} (1 - p) \rho(p) dp = 3 \lambda - \lambda m(t) - \lambda \widebar{m}(t),
    \end{equation*}
    which in turn means that the differential equation describing the evolution of $\widebar{m}(t)$ can be expressed as
    \begin{equation*}
        \frac{d \widebar{m} }{dt} = 3 \lambda - \lambda m(t) - \lambda \widebar{m}(t).
    \end{equation*}
    Equivalently,
    \begin{equation*}
        \frac{d}{dt} \left( e^{\lambda t} \widebar{m}(t) \right) = \lambda e^{\lambda t} (3 - m(t)).
    \end{equation*}
    Integrating gives the claim.
\end{proof}

The fact that this construction approximately implements continuous negation is formalized in the following lemma, whose proof appears in~\Cref{sec:proofsChua} (\Cref{fig:lam} contains an illustration).

\begin{restatable}{lemma}{approxint}
    \label{lemma:approxint}
    In the context of~\Cref{lemma:continuousnegation}, if in addition $\widebar{m}(0) = 3 - m(0)$, then
    \begin{equation*}
        | \widebar{m}(t) - (3 - m(t))| \leq \Theta_\lambda \left( \frac{\ln \lambda}{\lambda} \right) \quad \forall t \geq 0.
    \end{equation*}
\end{restatable}

Without the assumption regarding initialization, the continuous negation is still valid, but only for sufficiently large $t$.

\begin{figure}
    \centering
    \includegraphics[scale=0.36]{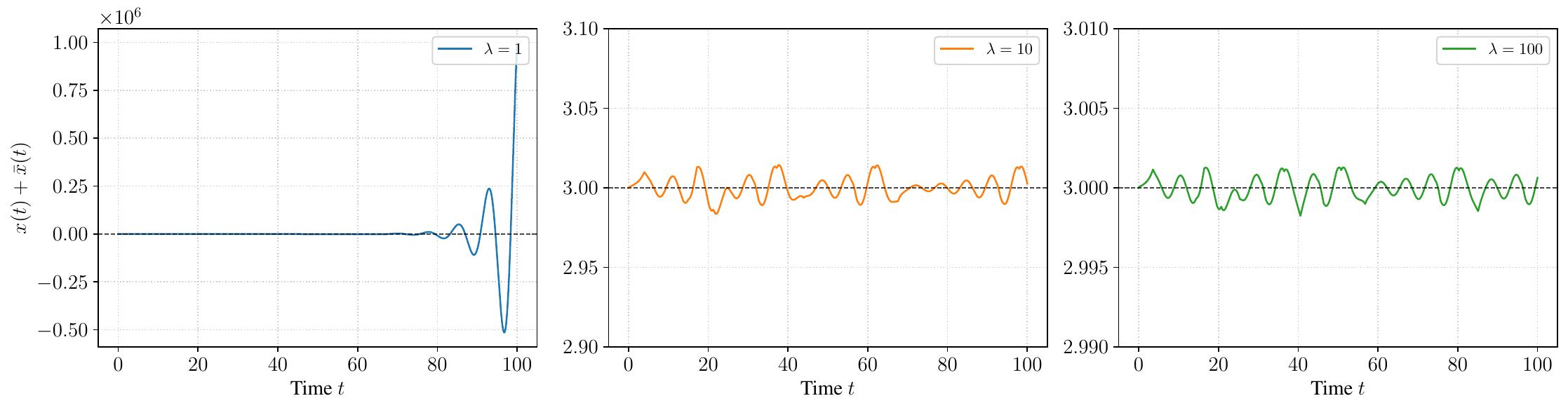}
    \caption{\Cref{lemma:continuousnegation} in action concerning~\eqref{eq:augmented-Chua}. As $\lambda$ increases, $x(t) + \protect\widebar{x}(t)$ approaches 3.}
    \label{fig:lam}
\end{figure}

\subsubsection{Putting everything together}

Armed with~\Cref{lemma:nonlinear-sim,lemma:continuousnegation}, we now show how to (approximately) simulate any nonlinear continuous dynamical system of the form given in~\eqref{eq:nonlinear}. Let us first formalize our notion of approximate simulation.

\begin{definition}[Approximate simulation]
    \label{def:approx-sim}
    Consider a $d$-dimensional continuous dynamical system $\Sigma$ described by the vector field $F : \R^d \to \R^d$. We say that a continuous dynamical system $\hat{\Sigma}$, given by $\frac{d \vx}{dt} = \hat{F}(\vx)$, \emph{$\epsilon$-simulates} $\Sigma$ if $\| F(\vx(t)) - \hat{F}(\vx(t)) \| \leq \epsilon$ for any $t \geq 0$.
\end{definition}

We should caveat this definition by pointing out that the fact that the vector field is approximated to some small accuracy does not necessarily imply that the orbits of the two systems remain close. For example, if $\Sigma$ is a stable linear system right at the edge of stability, a small perturbation in $F$ can result in an unstable $\hat{\Sigma}$; so, the trajectories of the two systems diverge at an exponential rate. Yet, \Cref{def:approx-sim} is often enough to capture interesting properties of the original system, as we shall demonstrate in the sequel. 

We further note that near a \emph{hyperbolic} invariant set, $\epsilon$-orbits (obtained under the approximate vector field $\hat{F}$) \emph{do} remain within a neighborhood of some true orbit; this is known as the \emph{shadowing lemma}~\citep{Hammel88:Numerical}. The precondition of the shadowing lemma---the presence of a hyperbolic invariant set---is not satisfied for Chua's circuit, which is why we will need to carry out a deeper analysis to examine the behavior of its approximation (\Cref{sec:pos-topentr,sec:posLLE}).

In this context, we establish the following theorem.

\begin{restatable}{theorem}{simgen}
    \label{theorem:simulation-general}
    Consider any continuous dynamical system $\Sigma$ per~\eqref{eq:nonlinear} satisfying~\Cref{assumption:nonlinear}. There exists an autobidding system such that a projection of the dynamics induced by~\eqref{eq:continuous-time} $\epsilon$-simulates $\Sigma$ (\Cref{def:approx-sim}).
\end{restatable}

By ``projection of the dynamics'' above we simply mean a restriction into a subset of its coordinates.

The proof proceeds by first converting the linear component of the system into a \emph{competitive} linear system. This is done using the continuous negation gadget (\Cref{lemma:continuousnegation}). To simulate the competitive linear system, we make use of certain gadgets developed by~\citet{Leme24:Complex}, while the nonlinear component can be recovered using~\Cref{lemma:nonlinear-sim}. The overall structure of the reduction was shown earlier in~\Cref{fig:reduction}. We defer the detailed proof to~\Cref{sec:proofsChua}.

It is worth comparing our simulation approach with that of~\citet{Leme24:Complex} concerning linear dynamical systems. While their notion of simulation is exact, it requires applying in addition a linear transformation. On the other hand, our simulation is with respect to a certain \emph{projection}---that is, a subset of the multipliers---of the autobidding system. This is an important distinction in linear systems, and turns out to be crucial for simulating nonlinear systems. Specifically, if one applies a linear transformation to a system as in~\eqref{eq:nonlinear}, one ends up with a nonlinearities that contain \emph{coupled} terms of the form $h_i(x_i, x_{i'}, \dots,)$. As we have stressed already, it is unclear how to recover such terms through the use of autobidding dynamics.

\subsection{Chua's circuit}\label{sec:chua}

Chua's circuit is a landmark continuous-time system whose central role in dynamical systems was highlighted earlier in~\Cref{sec:related}. The state equations are as follows.

\begin{equation}\label{eq:Chua}\begin{gathered}C_1 \frac{dx}{dt} = G(y - x) - g(x),\\ C_2 \frac{dy}{dt} = G(x - y) + z, \\
L \frac{dz}{dt} = -y - R_0 z.
\end{gathered}\end{equation}
Above, $g(x)$ is a three-segment piecewise-linear function representing the characteristic curve of Chua's diode, while $C_1$, $C_2$, $G$, $R_0$, and $L$ are parameters of the system. We defer the physical interpretation together with a precise specification of its parameters in~\Cref{appendix:Chuacircuit}. What is important for our purposes here is that Chua's circuit adheres to the structure imposed in~\Cref{sec:sim-general}. In particular, the only nonlinear component appears in the evolution of $x$ in~\eqref{eq:Chua}, and crucially is a function of $x$ itself. For the set of parameters we consider it also has bounded orbits, so we can apply our general simulation result.


More concretely, let us examine how our reduction proceeds in this case. For the specific set of parameters we consider (\Cref{appendix:Chuacircuit}), we have $x(t) \in [-2.5, 2.5]$, $y(t) \in [-0.45, 0.45]$, and $z(t) \in [-9.5, 9.5]$. The first step is to rescale and shift each variable such that the system remains bounded in $[1.1, 1.9]^3$. We will not spell out this step here, as it complicates the ensuing expressions, but is accounted for in our simulations. The more interesting step first involves making sure that the nonlinearity $\tilde{g}$ has a negative derivative (as required by~\Cref{lemma:nonlinear-sim}), which for the specific set of parameters we consider can be achieved by adding and subtracting $x$ in the first differential equation. And second, the use of continuous negation (\Cref{lemma:continuousnegation}) to create a competitive linear system. These steps are summarized below (without accounting for shifting and rescaling).

\begin{equation}
\label{eq:augmented-Chua}
\begin{aligned}
    C_1 \frac{dx}{dt} &= 3 - \widebar{x} + G(3 - \widebar{y}) + \tilde{g}(x), & \qquad 
    \frac{d \widebar{x}}{dt} &= 3 \lambda - \lambda x - \lambda \widebar{x}, \\
    C_2 \frac{dy}{dt} &= G (3 - \widebar{x}) - G y + (3 - \widebar{z}), & \qquad 
    \frac{d \widebar{y}}{dt} &= 3 \lambda - \lambda y - \lambda \widebar{y}, \\
    L \frac{dz}{dt} &= -y - R_0 z, & \qquad 
    \frac{d \widebar{z}}{dt} &= 3 \lambda - \lambda z - \lambda \widebar{z}.
\end{aligned}
\end{equation}

The validity of the simulation rests on using a large enough $\lambda > 0$. \Cref{fig:lam} shows the effect of $\lambda$ on $x(t) + \widebar{x}(t)$. We find that, as expected, small values of $\lambda$ break the consistency of the simulation. Visually, for larger values of $\lambda$ the orbits closely match the iconic double scroll attractor of Chua's circuit. In particular, \Cref{fig:autobidding-chaos} from our introduction corresponds to the rescaled and shifted version of~\eqref{eq:augmented-Chua} under $\lambda = 100$.

\subsubsection{Positive Lyapunov exponent}
\label{sec:posLLE}

Beyond visual confirmation, we dive deeper into the properties of the autobidding system produced by our reduction. First, we estimate the largest \emph{Lyapunov exponent}, which measures the rate of exponential separation; \Cref{appendix:LLE} provides further background together with a description of the standard way in which we performed the estimation. 

We find that the largest Lyapunov exponent of our autobidding system is $\approx 0.14214$. This is obtained again using $\lambda = 100$. The existence of a positive Lyapunov exponent in a bounded system provides strong confirmation for the onset of chaos.

\subsubsection{Positive topological entropy}
\label{sec:pos-topentr}

Finally, we examine whether the approximate version of Chua's circuit given in~\eqref{eq:augmented-Chua} has positive \emph{topological entropy} (\Cref{sec:top-entropy}). For Chua's circuit, \citet{Galias97:Positive} provided the first rigorous proof of positive topological entropy through a computer-assisted proof, in particular through the use of interval arithmetic. To analyze~\eqref{eq:augmented-Chua}, we follow the approach of~\citet{Galias97:Positive}. The overall framework is quite elaborate, and its presentation is deferred to~\Cref{app:galias_framework}. The key step in which Galias made use of interval arithmetic has to with how certain quadrangles map under the \emph{Poincar\'e map} (\Cref{sec:Poincare}) induced by Chua's circuit. To conclude the existence of infinitely many periodic orbits, one needs to guarantee a specific geometric property (\Cref{theorem:deformed}). We show that for a sufficiently large $\lambda$, \eqref{eq:augmented-Chua} indeed satisfies the required geometric property (\Cref{fig:quadrangles-approx}). Interestingly, for this to be the case, $\lambda$ needs to be considerably larger than what used in~\Cref{fig:autobidding-chaos} and the Lyapunov exponent calculation.

\begin{figure}[t]
    \centering
    \includegraphics[scale=0.42]{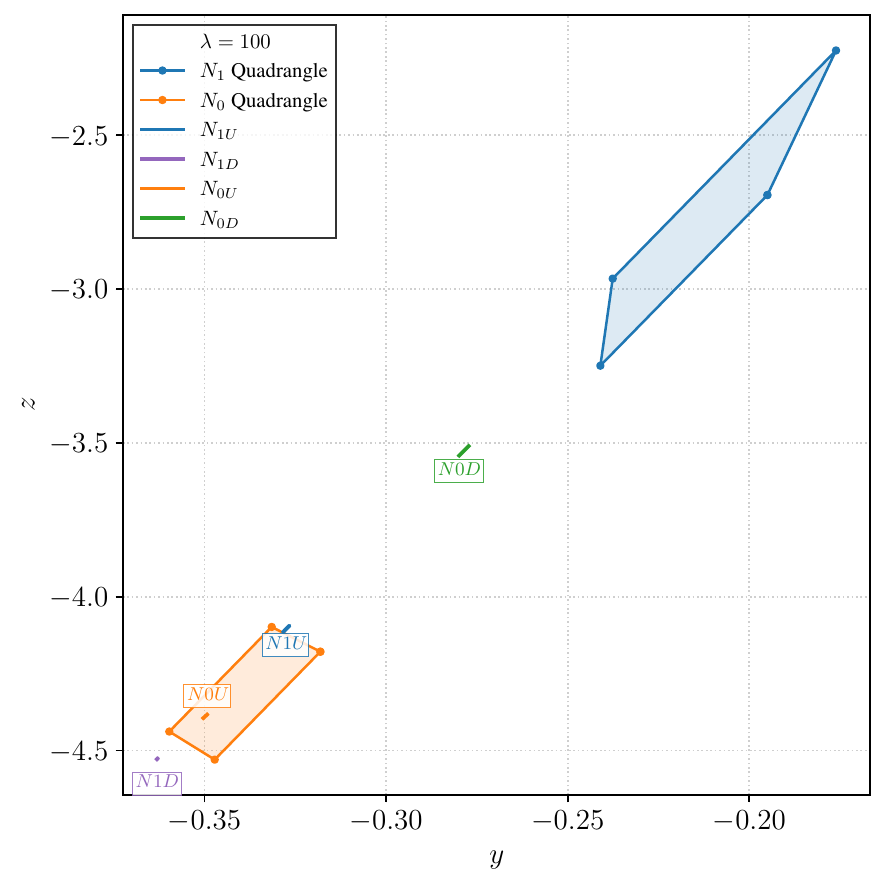}
    \includegraphics[scale=0.42]{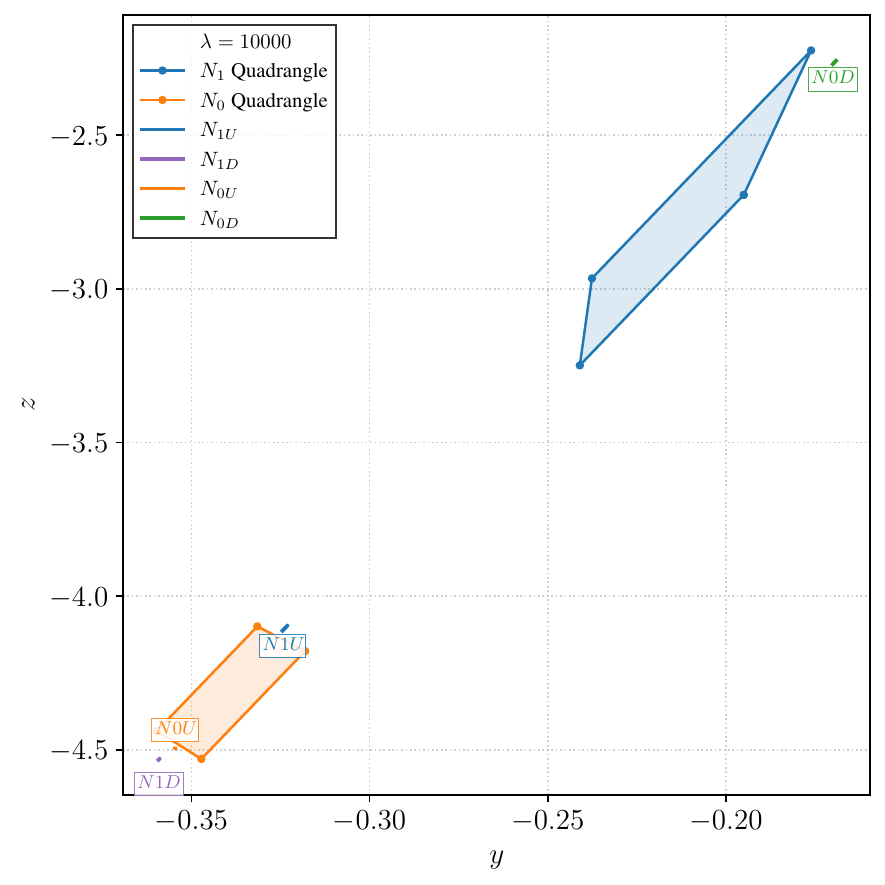}
    \caption{Verification that the approximate version of Chua's circuit given in~\eqref{eq:augmented-Chua} satisfies the precondition of~\Cref{theorem:deformed}. A detailed, self-contained overview of the significance of this geometric condition and its relation to the iconic (deformed) horseshoe map is given in~\Cref{app:galias_framework}.}
    \label{fig:quadrangles-approx}
\end{figure}

While we do not carry out the rigorous verification using interval arithmetic, \Cref{fig:quadrangles-approx} provides further strong evidence that, when $\lambda$ is large enough, the continuous negation gadget does not qualitatively alter the behavior of the system.

\section{Chaos in discrete-time autobidding dynamics}
\label{sec:discrete}

In the second part of the paper, we analyze the behavior of discrete-time autobidding dynamics. We first focus on a particular discrete-time version of the previously analyzed continuous-time dynamics~\eqref{eq:continuous-time}, namely
\begin{equation}
    \label{eq:expmap}
    \vec{m}^{(t+1)} \defeq \vec{m}^{(t)} \circ \exp( \eta u(\vec{m}^{(t)}) ),
\end{equation}
for some \emph{learning rate} $\eta > 0$. As we shall show formally, the learning rate dictates the behavior of the dynamics. Above, we use the shorthand notation $\vm = (m_1, \dots, m_n)$ and 
\begin{equation*}
    \exp (\eta u(\vm^{(t)})) = (\exp (\eta u_1( \vm^{(t)} ), \dots, \exp(\eta u_n(\vm^{(t)})),
\end{equation*}
where $\circ$ is the component-wise product. The dynamics~\eqref{eq:expmap} produced by the exponential map can be viewed as an instance of the classic \emph{mirror descent} algorithm~\citep{NemirovskiYudin83:Problem} under the entropic regularizer.

\begin{claim}[Entropic mirror descent]
    \label{claim:entropic}
    The update rule~\eqref{eq:expmap} can be equivalently expressed as $m_i^{(t+1)} \in \argmax_{m_i \geq 0} \left\{ \eta m_i u_i(\vm^{(t)}) - m_i \log m_i + m_i \log m_i^{(t)} + m_i - m_i^{(t)} \right\} $ for all autobidders $i \in [n]$.
\end{claim}

We also examine the usual gradient descent algorithm, which can again be viewed as mirror descent but with \emph{Euclidean} regularization:
\begin{equation}
    \label{eq:GD}
    m_i^{(t+1)} = \Pi_{\R_{\geq 0}} ( m_i^{(t)} + \eta u_i(\vm^{(t)})) = \max\{0, m_i^{(t)} + \eta u_i(\vm^{(t)}) \} \quad \forall i \in [n].
\end{equation}

Although the utility function of each autobidder may not be continuous, we will restrict our attention to dynamical systems that are continuous. It is worth noting that continuity in the utility functions can always be enforced by assuming a continuum of items.

We begin by analyzing the discrete-time dynamical system~\eqref{eq:expmap} in the setting of two autobidders with symmetric valuation profiles (\Cref{sec:symmetry}). Even in this simple setting, we establish that, when the learning rate is large enough, it can exhibit \emph{Li-Yorke chaos} (\Cref{theorem:Li-Yorke-auto}); this holds for an infinite parametric class of (symmetric) valuation profiles. \Cref{sec:Ricker} then makes a more general connection: under a suitable valuation profile, the system~\eqref{eq:expmap} can exactly recover the classic Ricker model (\Cref{thm:recover-Ricker}). In fact, the upshot is that the symmetric system we analyze in~\Cref{sec:symmetry} is just an instance of the Ricker model. We go on to show that gradient descent on a suitably constructed instance recovers the classic \emph{logistic map} (\Cref{theorem:logistic}). The underlying market is particularly natural and simple. Just like~\Cref{sec:Chua}, these connections allows us to inherit many theoretical insights from dynamical systems to the emerging study of autobidding dynamics.

\subsection{Symmetric systems}
\label{sec:symmetry}

We first turn to \emph{symmetric} instances, by which we mean problems in which $u_i(m, \dots, m) = u_{i'}(m, \dots, m)$ for any $m > 0$ and $i, i' \in [n]$. For such systems, it follows inductively that if we initialize at a point of the form $(m, m \dots, m)$, we maintain the invariance $m^{(t)}_{i} = m^{(t)}_{i'}$ for all $i, i' \in [n]$ under~\eqref{eq:expmap}. One consequence of this is that the resulting system is essentially one-dimensional, making it more analytically tractable.

In fact, we will show that chaos emerges even in the simple case of two autobidders and two items. In particular, we consider a parametric family of valuations given by
\begin{equation}
    \label{eq:param-valuations}
    \mat{V} \defeq \begin{bmatrix}
        v  & 1 \\
        1 & v
    \end{bmatrix}
\end{equation}
for some parameter $v > 1$. In words, each autobidder strictly prefers their respective item---the first autobidder favors the first item and the second autobidder favors the second item; as $v \gg 1$, this gap grows larger. It is clear that this instance adheres to the notion of symmetry defined above: when they both employ the same multiplier, each secures their preferred item valued at $v$ and for a price equal to $m$. Specifically, the one-dimensional map under the exponential dynamics~\eqref{eq:expmap} is given by $F(m) = m e^{\eta (v - m)}$.

It is worth pointing out that this simple instance already reveals fundamental differences between utility maximization and value maximization subject to RoS constraints. Specifically, under a suitable tie breaking rule, there are equilibria in which both items are procured by the same autobidder. When $v \gg 1$, such an equilibrium delivers only about $50 \%$ of the optimal welfare, a price of anarchy bound which is tight~\citep{Aggarwal19:Autobidding}.

From a dynamical point of view, we find that such instances can support a broad range of behaviors depending on the learning rate. Pictorially, \Cref{fig:cobweb} portrays the \emph{cobweb plots} for different values of $\eta$. The plots illustrate the system's evolution as the learning rate increases: the transition from asymptotic stability (left), to stable limit cycles (middle), and finally to complex aperiodicity (right). We will formalize this behavior, showing that, for a large enough learning rate, the system exhibits \emph{Li-Yorke chaos} (\Cref{theorem:Li-Yorke-auto}). A neat representation of this transition can be found in the \emph{bifurcation diagram} of the system, presented in~\Cref{sec:bifurcation}.

Before we do so, let us first comment on the other two regimes. First, when $\eta$ is small enough so that $| 1 - \eta v| < 1$, the unique fixed point---which is $v$---is \emph{linearly stable}~\citep{Strogatz94:Nonlinear}, which implies local asymptotic stability; this follows by inspecting the absolute value of the derivative of the map---which is differentiable for the particular system we consider---at the fixed point. This agrees with the analysis of~\citet{Leme24:Complex} in continuous time, where they proved global convergence through an application of the Poincar\'e-Bendixson theorem. In the intermediate regime in terms of the learning rate, the system ends up at a limit cycle (\Cref{fig:cobweb}, middle). In particular, when $\eta = \ln 3 \approx 1.0986$, the system (with $v = 2$) has the period-$2$ orbit $1 \to e^{\ln 3} = 3 \to 1$; for other values of $\eta$ in that neighborhood, the system still ends up at a period-$2$ orbit, as is also clearly illustrated in the bifurcation diagram (\Cref{fig:bifurcation}).

\begin{figure}[t]
    \centering
    \includegraphics[scale=0.5]{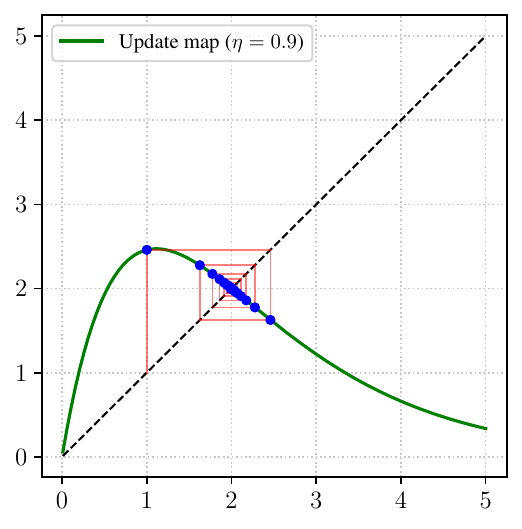}
    \includegraphics[scale=0.5]{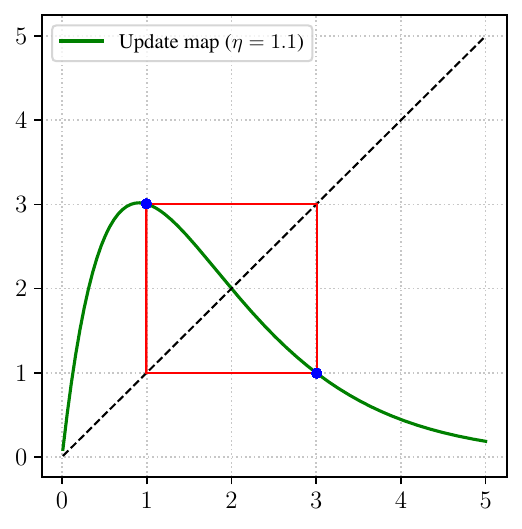}
    \includegraphics[scale=0.5]{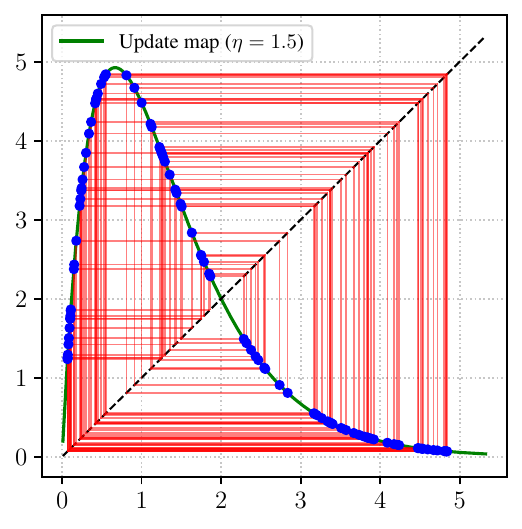}
    \caption{Cobweb plots for different values of the learning $\eta \in \{0.9, 1.1, 1.5\}$ under the dynamics~\eqref{eq:expmap} and the valuations given in~\eqref{eq:param-valuations} for $v = 2$. We refer to~\Cref{sec:cobweb} as to how the plots are produced.}
    \label{fig:cobweb}
\end{figure}

We now formalize that for sufficiently large $\eta$, the system exhibits Li-Yorke chaos. To do so, it suffices to establish the existence of orbits with period $3$, as formalized below.

\begin{restatable}[Period 3]{lemma}{periodthree}
    \label{lemma:period3}
    Let $F$ be the map associated with ~\eqref{eq:expmap} under the valuation profile given in~\eqref{eq:param-valuations}. For any $\eta \geq \max \{ 1, \frac{1}{v-1} \ln (3v - 1) \}$, 
    there exists a point $m \in (1, v)$ such that $F^{(3)}(m) = m$, $F(m) \neq m$, and $F^{(2)}(m) \neq F(m)$.
\end{restatable}

The proof is deferred to~\Cref{sec:proofs-discrete}. Combining with the famous theorem of~\citet{Li75:Period}, we arrive at the following result.

\begin{theorem}[Autobidding dynamics can exhibit Li-Yorke chaos]
    \label{theorem:Li-Yorke-auto}
    There is an infinite class of symmetric autobidding instances in which the dynamical system~\eqref{eq:expmap} exhibits Li-Yorke chaos.
\end{theorem}

Another interesting aspect of this dynamical system is that it has negative \emph{Schwarzian derivative}; for a $C^3$ map $F$, the Schwarzian derivative, $SF(m)$, is defined as 
\[
SF \defeq \frac{F'''}{F'} - \frac{3}{2} \left( \frac{F''}{F'} \right)^2.
\]
A simple calculation reveals the following.

\begin{restatable}{lemma}{Schwarzian}
    For any learning rate $\eta > 0$, the dynamical system~\eqref{eq:expmap} under the valuation profile given in~\eqref{eq:param-valuations} has negative Schwarzian derivative, $SF(m)$, for any $m \geq 0$.
\end{restatable}

This imposes a strict regularity on the system, following a universal period-doubling route to chaos~\citep{Feigenbaum78:Quantitative} characteristic of \emph{S-unimodal maps} (\Cref{def:unimodal}). In particular, by Singer's theorem~\citep{Singer78:Stable}, the system admits at most one stable attracting periodic point (excluding boundary fixed points), precluding multistability. This guarantees that as the learning rate $\eta$ increases, the system follows a standard period-doubling route to chaos without the coexistence of competing stable cycles, thereby ensuring a structurally predictable transition into the chaotic regime (\Cref{fig:bifurcation}).

Interestingly, despite the emergence of chaos, we find empirically that the welfare and \emph{average} revenue of the system both converge to a value of $4$ when $v = 2$ (\Cref{fig:revenuewelfare}), matching the outcomes observed in the convergent regime (\Cref{fig:revenuewelfare}, left); a welfare and revenue of $4$ correspond to each autobidder winning their preferred item but at a price equal to their value for it. This shows that, in this example, the macroscopic behavior is robust to dynamic instability, as the time-averaged statistics smooth out the volatility of individual bidding strategies. In stark contrast, a system analyzed in the sequel (\Cref{fig:logistic-revenuewelfare}) shows how macroscopic performance can degrade in the chaotic regime.

It is worth noting that the chaos documented above is eliminated if one restricts the multiplier to be at least $1$; that is, $\vm^{(t+1)} = \max\{1, \vec{m}^{(t)} \circ \exp( \eta u(\vec{m}^{(t)}) ) \}$. This truncation is quite natural since choosing a multiplier strictly less than $1$ is a (weakly) dominated strategy. Even when the learning rate is large, these truncated dynamics avoid chaos and settle into a stable period-$2$ orbit---at least under the valuation profile given in~\eqref{eq:param-valuations}. This is another sense in which the chaos established in~\Cref{sec:Chua} is more robust.

In a similar vein, let us comment on mirror descent but with Euclidean regularization. Under the valuation profile given in~\eqref{eq:param-valuations}, it follows that $m^{(t+1)} = \max\{0, m^{(t)} + \eta (v - m^{(t)} ) \}$. Numerical simulations confirm that this map lacks the mechanism necessary for chaotic mixing and cannot support the complex behavior observed under~\eqref{eq:expmap}. In stark contrast, the next section provides an example where the Euclidean regularizer does produce chaos.

\begin{figure}[t]
    \centering
    \includegraphics[scale=0.4]{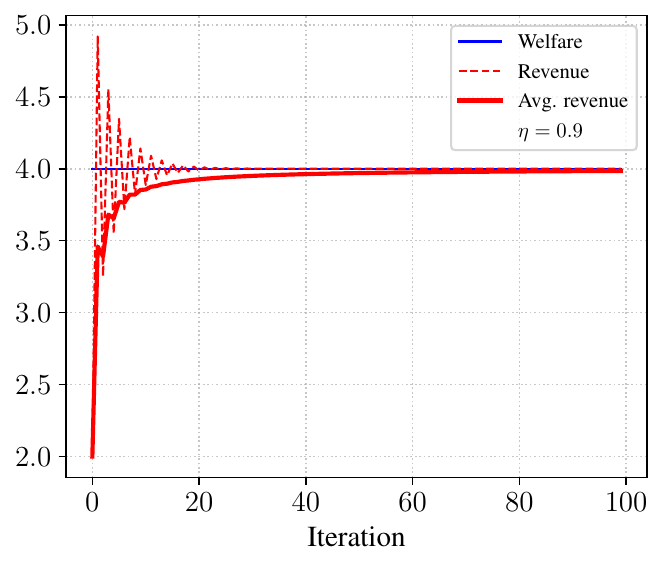}
    \includegraphics[scale=0.4]{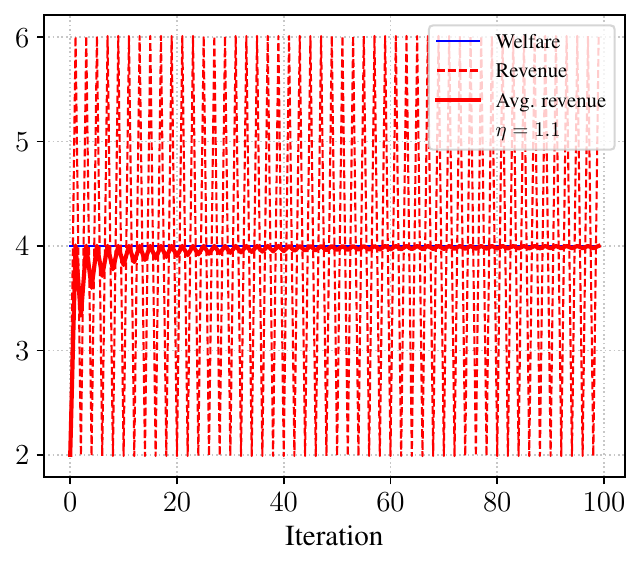}
    \includegraphics[scale=0.4]{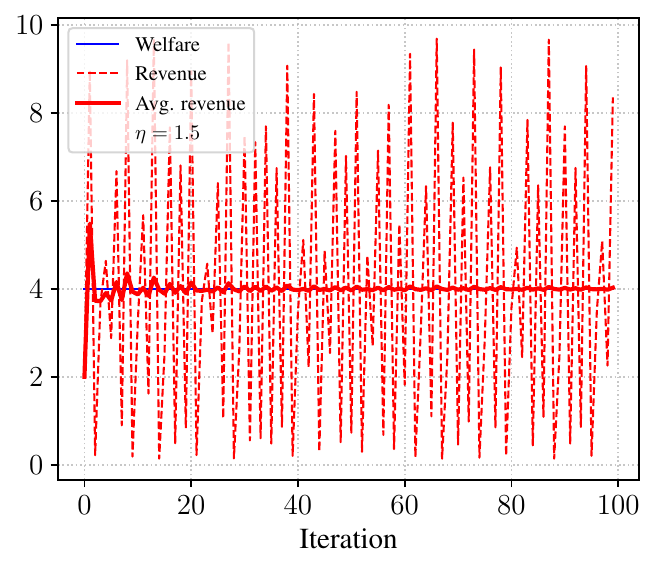}
    \caption{Welfare and revenue under the dynamics~\eqref{eq:expmap} and the valuations given in~\eqref{eq:param-valuations} for $v = 2$.}
    \label{fig:revenuewelfare}
\end{figure}

\subsection{Recovering classic dynamical systems}
\label{sec:Ricker}

Li-Yorke chaos is just one possible manifestation of chaos. As in~\Cref{sec:Chua}, our goal is to provide a more precise characterization by resorting to existing results. Also, one limitation of the previous result is that the advent of Li-Yorke chaos hinges on symmetric initialization, making the ensuing behavior rather brittle. In this subsection, we first observe that one can recover the well-known \emph{Ricker model}~\citep{Ricker54:Stock} via suitable autobidding instances, and this is so without requiring symmetry. In fact, the upshot is that the dynamical system analyzed in~\Cref{sec:symmetry} is itself an instance of the Ricker model. We then show that autobidding dynamics can recover the logistic map.

\paragraph{Ricker model} The Ricker population model was introduced in the context of stock and recruitment in fisheries. It describes the expected number in a population $N^{(t)}$ at time $t$. The evolution of the system is governed by the recursion
\begin{equation*}
    N^{(t+1)} = N^{(t)} \exp\left( r \left(1 - \frac{N^{(t)}}{k} \right) \right),
\end{equation*}
where $r$ is the intrinsic growth rate and $k$ the \emph{carrying capacity} of the environment---the maximum population size that can be sustained. It is a well-known fact that the Ricker model exhibits Li-Yorke chaos~\citep{Elaydi07:Discrete}. It can also be shown to exhibit other forms of chaos, such as \emph{topological transitivity} (\Cref{def:transitivity}) and \emph{sensitivity to initial conditions} (\Cref{def:sensitivity}). 

We observe that by interpreting $r$ as the learning rate and setting $\mat{V} = [[ 1, 1/k], [1/k, 1]]$ with $k > 1$, the dynamics~\eqref{eq:expmap} exactly recover the Ricker model.
\begin{theorem}[Autobidding dynamics subsume the Ricker model]
\label{thm:recover-Ricker}
    There is an autobidding instance in which the dynamics~\eqref{eq:expmap} recover the Ricker model.
\end{theorem}
This connection allows us to inherit existing results from the dynamical systems literature to the autobidding setting. We highlight a notable implication below. (Detailed background concerning the relevant concepts from dynamical systems is given in~\Cref{sec:moreprels}.)

\begin{theorem}
    \label{thm:Devaney}
    There is an infinite class of autobidding instances in which the dynamical system~\eqref{eq:expmap} exhibits Devaney chaos on an attractor $\Lambda$. In particular, the system is sensitive to initial conditions.
\end{theorem}

Sensitivity to initial conditions is a particularly intuitive manifestation of chaos, popularly known as ``butterfly effects.'' From a practical standpoint, it hinders numerical simulation since a tiny error---caused, for example, by floating-point arithmetic or measurement noise---gets exponentially amplified over time, rendering prediction impossible beyond a short time frame.

\begin{figure}[t]
    \centering
    \includegraphics[scale=0.4]{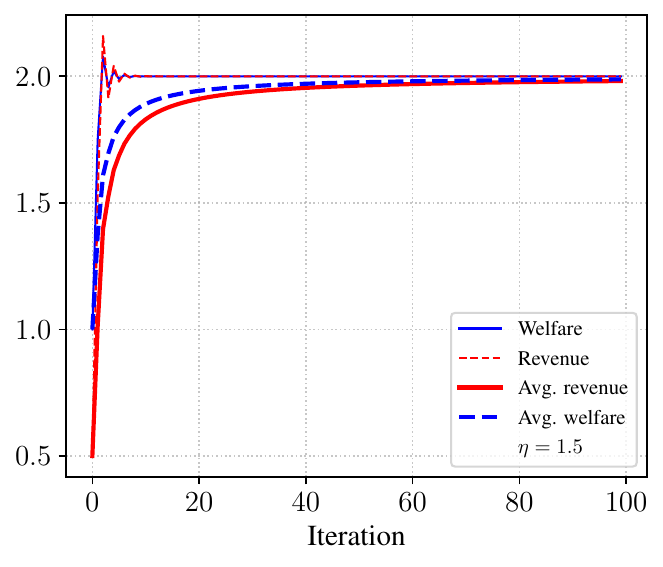}
    \includegraphics[scale=0.4]{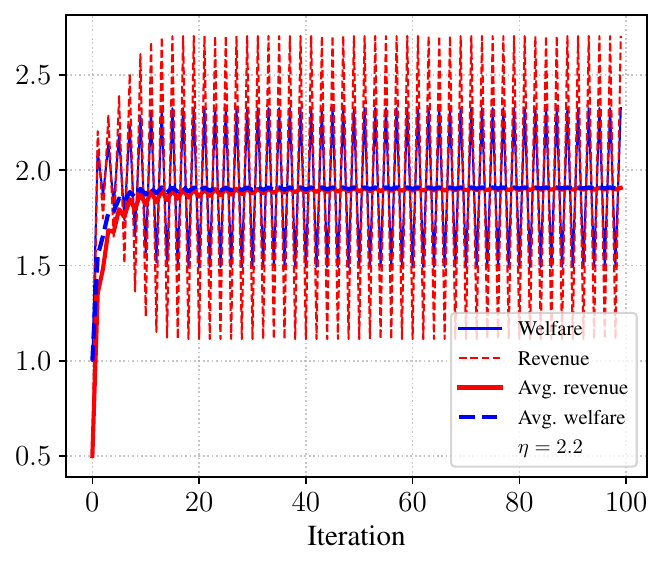}
    \includegraphics[scale=0.4]{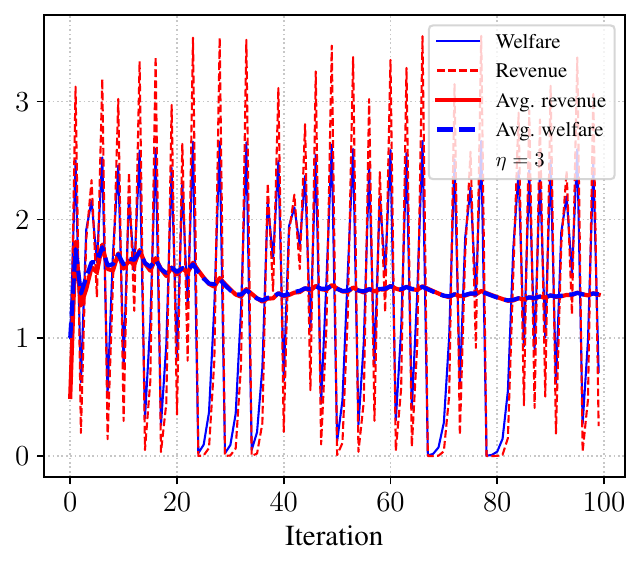}
    \caption{Welfare and revenue under the dynamics~\eqref{eq:GD} corresponding to~\Cref{theorem:logistic}.}
    \label{fig:logistic-revenuewelfare}
\end{figure}

\paragraph{Gradient descent and the logistic map} The foregoing results concern the mirror descent dynamics induced by the entropic regularizer. We now characterize the complexity of mirror descent but with Euclidean regularization, also known as gradient descent. Under a single autobidder, the dynamical system takes the form $m^{(t+1)} = \Pi_{\R_{\geq 0}} (m^{(t)} + \eta u(m^{(t)})) = \max\{0, m^{(t)} + \eta u(m^{(t)}) \}$. We observe that gradient descent in the autobidding setting can recover the famous \emph{logistic map}, one of the most well-studied dynamical systems. In particular, similarly to~\Cref{lemma:nonlinear-sim}, we first show how to create an autobidder whose utility matches any function satisfying certain properties.

\begin{lemma}
    \label{lemma:sim-max}
    Let $h : [0, \infty) \to \R$ be any differentiable function that is increasing in $[0, 1)$ and decreasing in $(1, \infty)$. Suppose further that $h(0) = 0$ and $h''(1)$ exists. There is an instance in which the utility of the autobidder is equal to $h$.
\end{lemma}

\begin{proof}
    We consider a continuum of items represented using a density function $\rho : [0, \infty) \to \R_{\geq 0}$ parameterized on the reserve price. Each item is valued equally for $1$ per unit. We set $\rho(p) = \frac{h'(p)}{1 - p}$ for $p \neq 1$, which is nonnegative since, by assumption, $h'(p) > 0$ for $p \in (0, 1)$ and $h'(p) < 0$ for $p > 1$. Further, $\rho(1) = - h''(1)$. The utility of the autobidder under a multiplier $m \geq 0$ is equal to $\int_{p = 0}^m (1 - p) \rho(p) dp = h(m) - h(0) = h(m)$.
\end{proof}

The simplest example of this construction is when the density $\rho$ is constant (one can safely impose an upper bound on the density if the multiplier remains bounded), in which case the utility function produced is the quadratic $\int_{p = 0}^m (1 - p) = m - \frac{m^2}{2}$. We observe that executing gradient descent on this instance recovers a rescaled version of the logistic map. Specifically, if $\eta = r - 1$, we have
\begin{equation}
    \label{eq:rescaled-log}
    F(m) = \max\Bigg\{0, m + (\overbrace{r-1}^{\eta}) \overbrace{\left( m - \frac{m^2}{2} \right)}^{u(m)} \Bigg\} = \max\left\{0, r m \left( 1 - \frac{r-1}{2r} m \right) \right\}.
\end{equation}
If $m \in (0, \frac{2r}{r-1})$ and $r \leq 4$, it follows that $r m \left( 1 - \frac{r-1}{2r} m \right) \in (0, \frac{2r}{r-1})$. So, in that regime,
\begin{equation*}
    F(m) = r m \left( 1 - \frac{r-1}{2r} m \right).
\end{equation*}
If we keep track of the rescaled multiplier $x = \frac{r-1}{2r} m$, we find that $x^{(t+1)} = r x^{(t)} (1 - x^{(t)})$, which is precisely the logistic map. The bifurcation diagram under~\eqref{eq:rescaled-log} is given in~\Cref{fig:GD-bifurcation}.

\begin{theorem}[Autobidding dynamics subsume the logistic map]
    \label{theorem:logistic}
    There is an instance with a single autobidder and reserve prices such that the dynamics~\eqref{eq:GD} recover the logistic map.
\end{theorem}

\Cref{fig:logistic-revenuewelfare} portrays the welfare and revenue corresponding to~\Cref{theorem:logistic}. In the convergent regime (\Cref{fig:logistic-revenuewelfare}, left), we find that the average welfare and revenue converge to a value of $2$, which corresponds to the unique autobidding equilibrium of this instance. Even when the dynamics are cycling (\Cref{fig:logistic-revenuewelfare}, middle), the average welfare and revenue are approaching values close to $2$, although they are not converging to it. In stark contrast, in the chaotic regime (\Cref{fig:logistic-revenuewelfare}, right), both average revenue and welfare are highly suboptimal. This suggests another critical implication of chaos: besides the erratic nature of the dynamics, macroscopic performance can also degrade.

Based on~\Cref{lemma:sim-max}, one can also consider the entropic mirror descent dynamics when $u(m) = m - \frac{m^2}{2}$, namely
\[
    m^{(t+1)} = m^{(t)} \exp\left( \eta \left( m^{(t)} - \frac{(m^{(t)})^2}{2} \right) \right).
\]
This is a variant of the Ricker map with a quadratic term. \Cref{fig:newcomp} confirms that when the learning rate is large enough, there is a period-3 orbit, and hence Li-Yorke chaos. We have thus identified an instance where Li-Yorke chaos manifests itself for both Euclidean and entropic regularization.

\begin{figure}
    \centering
    \includegraphics[scale=0.5]{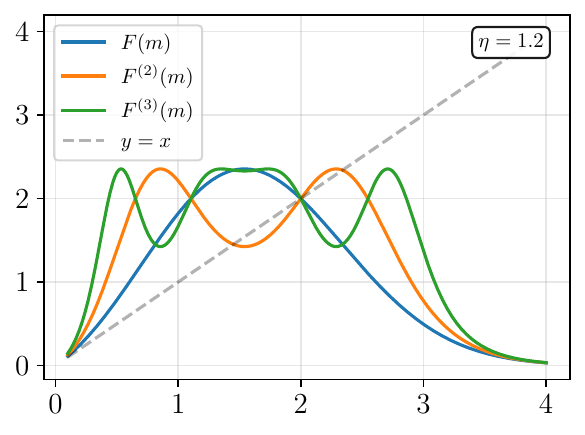}
    \hspace{1cm}
    \includegraphics[scale=0.5]{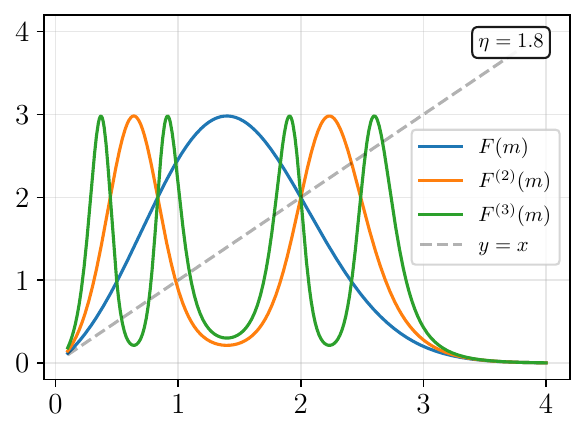}
    \caption{Compositions of $F(m) = m \exp(\eta( m - \frac{m^2}{2} ) )$. When $\eta = 1.8$ (right), there exists a period-3 orbit. On the other hand, when $\eta = 1.2$ (left), the only fixed point of $F^{(3)}$ is the unique fixed point of $F$.}
    \label{fig:newcomp}
\end{figure}
\section{Conclusions}

We established that the dynamics of value-maximizing autobidders under RoS constraints can be formally \emph{chaotic}, both in discrete and continuous time. For the more challenging continuous-time setting, our basic approach was to simulate existing, well-studied dynamical systems, such as Chua's circuit. For discrete-time dynamics, we connected different incarnations of mirror descent to classic dynamical systems such as the logistic map and the Ricker population model. The complex behavior established in discrete time manifests itself under remarkably simple market conditions, but, unlike our continuous-time results, hinges on autobidders employing a large learning rate.

The main takeaway of our results is that ostensibly simple autobidding dynamics can produce highly complex and unpredictable behavior. Chaos is not merely a theoretical curiosity; it has tangible practical implications. Most notably, the inherent sensitivity to initial conditions can cast doubt into the reliability of numerical simulations, making long-term market forecasting and simulation intractable. Going forward, more emphasis should be given beyond static equilibria toward statistical or ergodic properties of the dynamics.

\section*{Acknowledgments}

We are indebted to Zbigniew Galias for his valuable and prompt assistance concerning the simulation of~\Cref{fig:quadrangles-approx}.

\bibliography{refs}

@String{WINE = "International Workshop On Internet And Network Economics (WINE)"}

@String{IJCAI = "Proceedings of the International Joint
                 Conference on Artificial Intelligence (IJCAI)"}

@String{AAAI = "Conference on Artificial Intelligence (AAAI)"}

@String{EC = "ACM Conference on Economics and Computation (EC)"}

@String{ICML = "International Conference on Machine Learning (ICML)"}

@String{WWW = "International Conference on World Wide Web (WWW)"}

@String{ITCS = "Innovations in Theoretical
                  Computer Science (ITCS)"}

@String{COLT = "Conference on Learning Theory (COLT)"}

@String{NeurIPS = "Neural Information Processing Systems (NeurIPS)"}

@inproceedings{Aggarwal19:Autobidding,
  author       = {Gagan Aggarwal and
                  Ashwinkumar Badanidiyuru and
                  Aranyak Mehta},
  title        = {Autobidding with Constraints},
  booktitle    = {Web and Internet Economics  (WINE)},
  year         = {2019}
}

@inproceedings{Balseiro21:Robust,
  author       = {Santiago R. Balseiro and
                  Yuan Deng and
                  Jieming Mao and
                  Vahab S. Mirrokni and
                  Song Zuo},
  title        = {Robust Auction Design in the Auto-bidding World},
  booktitle    = NeurIPS,
  year         = {2021}
}

@inproceedings{Deng24:Efficiency,
  author       = {Yuan Deng and
                  Jieming Mao and
                  Vahab Mirrokni and
                  Hanrui Zhang and
                  Song Zuo},
  title        = {Efficiency of the First-Price Auction in the Autobidding World},
  booktitle    = NeurIPS,
  year         = {2024}
}

@inproceedings{Fikioris23:Liquid,
  title={Liquid welfare guarantees for no-regret learning in sequential budgeted auctions},
  author={Fikioris, Giannis and Tardos, {\'E}va},
  booktitle=EC,
  year={2023}
}

@inproceedings{Baldeschi26:Optimal,
  title={Optimal Type-Dependent Liquid Welfare Guarantees for Autobidding Agents with Budgets},
  author={Baldeschi, Riccardo Colini and Klumper, Sophie and Kroll, Twan and Leonardi, Stefano and Sch{\"a}efer, Guido and Tsikiridis, Artem},
  booktitle={Symposium on Discrete Algorithms (SODA)},
  year={2026}
}

@article{Elowitz00:Synthetic,
  title={A synthetic oscillatory network of transcriptional regulators},
  author={Elowitz, Michael B and Leibler, Stanislas},
  journal={Nature},
  volume={403},
  number={6767},
  pages={335--338},
  year={2000}
}

@inproceedings{Bichler25:Beyond,
  author       = {Martin Bichler and
                  Stephan B. Lunowa and
                  Matthias Oberlechner and
                  Fabian R. Pieroth and
                  Barbara I. Wohlmuth},
  title        = {Beyond Monotonicity: On the Convergence of Learning Algorithms in
                  Standard Auction Games},
  booktitle    = AAAI,
  year         = {2025}
}

@inproceedings{Dobzinski14:Efficiency,
  title={Efficiency guarantees in auctions with budgets},
  author={Dobzinski, Shahar and Leme, Renato Paes},
  booktitle={International Colloquium on Automata, Languages, and Programming (ICALP)},
  year={2014}
}

@inproceedings{Deng21:Towards,
  author       = {Yuan Deng and
                  Jieming Mao and
                  Vahab S. Mirrokni and
                  Song Zuo},
  title        = {Towards Efficient Auctions in an Auto-bidding World},
  booktitle    = {The Web Conference (WWW)},
  year         = {2021}
}

@inproceedings{Leme24:Complex,
  author       = {Renato Paes Leme and
                  Georgios Piliouras and
                  Jon Schneider and
                  Kelly Spendlove and
                  Song Zuo},
  title        = {Complex Dynamics in Autobidding Systems},
  booktitle    = EC,
  year         = {2024}
}

@article{Li75:Period,
  title={Period three implies chaos},
  author={Li, Tien-Yien and Yorke, James A},
  journal={The American Mathematical Monthly},
  volume={82},
  number={10},
  pages={985--992},
  year={1975}
}

@article{Bruckner87:Scrambled,
  title={On scrambled sets for chaotic functions},
  author={Bruckner, Andrew M and Hu, Thakyin},
  journal={Transactions of the American Mathematical Society},
  volume={301},
  number={1},
  pages={289--297},
  year={1987}
}

@inproceedings{Papadimitriou16:Computational,
  author       = {Christos H. Papadimitriou and
                  Nisheeth K. Vishnoi},
  editor       = {Madhu Sudan},
  title        = {On the Computational Complexity of Limit Cycles in Dynamical Systems},
  booktitle    = ITCS,
  year         = {2016}
}

@inproceedings{Chatziafratis19:Computational,
  author       = {Vaggos Chatziafratis and
                  Tim Roughgarden and
                  Joshua R. Wang},
  title        = {On the Computational Power of Online Gradient Descent},
  booktitle    = COLT,
  year         = {2019}
}

@inproceedings{Gaitonde23:Budget,
  author       = {Jason Gaitonde and
                  Yingkai Li and
                  Bar Light and
                  Brendan Lucier and
                  Aleksandrs Slivkins},
  title        = {Budget Pacing in Repeated Auctions: Regret and Efficiency Without
                  Convergence},
  booktitle    = ITCS,
  year         = {2023}
}

@inproceedings{Lucier24:Autobidders,
  author       = {Brendan Lucier and
                  Sarath Pattathil and
                  Aleksandrs Slivkins and
                  Mengxiao Zhang},
  title        = {Autobidders with Budget and {ROI} Constraints: Efficiency, Regret,
                  and Pacing Dynamics},
  booktitle    = COLT,
  year         = {2024}
}

@article{Aggarwal24:Autobidding,
  author       = {Gagan Aggarwal and
                  Ashwinkumar Badanidiyuru and
                  Santiago R. Balseiro and
                  Kshipra Bhawalkar and
                  Yuan Deng and
                  Zhe Feng and
                  Gagan Goel and
                  Christopher Liaw and
                  Haihao Lu and
                  Mohammad Mahdian and
                  Jieming Mao and
                  Aranyak Mehta and
                  Vahab Mirrokni and
                  Renato Paes Leme and
                  Andr{\'{e}}s Perlroth and
                  Georgios Piliouras and
                  Jon Schneider and
                  Ariel Schvartzman and
                  Balasubramanian Sivan and
                  Kelly Spendlove and
                  Yifeng Teng and
                  Di Wang and
                  Hanrui Zhang and
                  Mingfei Zhao and
                  Wennan Zhu and
                  Song Zuo},
  title        = {Auto-Bidding and Auctions in Online Advertising: {A} Survey},
  journal      = {SIGecom Exchanges},
  year         = {2024}
}

@article{Ricker54:Stock,
  title={Stock and recruitment},
  author={Ricker, William Edwin},
  journal={Journal of the Fisheries Board of Canada},
  volume={11},
  number={5},
  pages={559--623},
  year={1954}
}

@book{Dold12:Lectures,
  title={Lectures on algebraic topology},
  author={Dold, Albrecht},
  volume={200},
  year={2012}
}

@article{Galias97:Positive,
  title={Positive topological entropy of Chua's circuit: A computer assisted proof},
  author={Galias, Zbigniew},
  journal={International Journal of Bifurcation and Chaos},
  volume={7},
  number={02},
  pages={331--349},
  year={1997}
}

@article{Tucker02:Rigorous,
  author  = {Tucker, Warwick},
  title   = {A Rigorous {ODE} Solver and {Smale}'s 14th Problem},
  journal = {Foundations of Computational Mathematics},
  year    = {2002},
  volume  = {2},
  number  = {1},
  pages   = {53--117}
}

@inproceedings{Li24:Vulnerabilities,
  author       = {Juncheng Li and
                  Pingzhong Tang},
  title        = {Vulnerabilities of Single-Round Incentive Compatibility in Auto-bidding:
                  Theory and Evidence from ROI-Constrained Online Advertising Markets},
  booktitle    = IJCAI,
  year         = {2024}
}

@article{Sharkovskii07:Co,
  title={Co-existence of cycles of a continuous mapping of the line into itself},
  author={Oleksandr Mykolayovych Sharkovskii},
  journal={Ukrainian Mathematical Journal},
  volume={16},
  pages={61--71},
  year={1964}
}

@article{Matsumoto03:Chaotic,
  author={Matsumoto, T.},
  journal={IEEE Transactions on Circuits and Systems}, 
  title={A chaotic attractor from Chua's circuit}, 
  year={1984},
  volume={31},
  number={12},
  pages={1055-1058}
}

@inproceedings{Palaiopanos17:Multiplicative,
  author       = {Gerasimos Palaiopanos and
                  Ioannis Panageas and
                  Georgios Piliouras},
  title        = {Multiplicative Weights Update with Constant Step-Size in Congestion
                  Games: Convergence, Limit Cycles and Chaos},
  booktitle    = {Neural Information Processing Systems (NeurIPS)},
  year         = {2017}
}

@inproceedings{Bielawski21:Follow,
  author       = {Jakub Bielawski and
                  Thiparat Chotibut and
                  Fryderyk Falniowski and
                  Grzegorz Kosiorowski and
                  Michal Misiurewicz and
                  Georgios Piliouras},
  title        = {Follow-the-Regularized-Leader Routes to Chaos in Routing Games},
  booktitle    = {International Conference on Machine Learning (ICML)},
  year         = {2021}
}

@inproceedings{Cheung20:Chaos,
  author       = {Yun Kuen Cheung and
                  Georgios Piliouras},
  title        = {Chaos, Extremism and Optimism: Volume Analysis of Learning in Games},
  booktitle    = {Neural Information Processing Systems (NeurIPS)},
  year         = {2020}
}

@article{Sato02:Chaos,
  title={Chaos in learning a simple two-person game},
  author={Sato, Yuzuru and Akiyama, Eizo and Farmer, J Doyne},
  journal={Proceedings of the National Academy of Sciences},
  volume={99},
  number={7},
  pages={4748--4751},
  year={2002}
}

@article{Morrison24:Diversity,
  title={Diversity of emergent dynamics in competitive threshold-linear networks},
  author={Morrison, Katherine and Degeratu, Anda and Itskov, Vladimir and Curto, Carina},
  journal={SIAM Journal on Applied Dynamical Systems},
  volume={23},
  number={1},
  pages={855--884},
  year={2024}
}

@article{Vickrey61:Counterspeculation,
  title={Counterspeculation, Auctions, and Competitive Sealed Tenders},
  author={Vickrey, William},
  journal={The Journal of Finance},
  volume={16},
  number={1},
  pages={8--37},
  year={1961}
}

@article{Clarke71:Multipart,
  title={Multipart Pricing of Public Goods},
  author={Clarke, Edward H},
  journal={Public Choice},
  volume={11},
  number={1},
  pages={17--33},
  year={1971}
}

@article{Groves73:Incentives,
  title={Incentives in Teams},
  author={Groves, Theodore},
  journal={Econometrica},
  volume={41},
  number={4},
  pages={617--631},
  year={1973}
}

@book{Strogatz94:Nonlinear,
  title={Nonlinear Dynamics and Chaos},
  author={Strogatz, Steven H.},
  year={1994},
  publisher={Perseus Books}
}

@article{Singer78:Stable,
  title={Stable orbits and bifurcation of maps of the interval},
  author={Singer, David},
  journal={SIAM Journal on Applied Mathematics},
  volume={35},
  number={2},
  pages={260--267},
  year={1978}
}

@article{Vellekoop94:Intervals,
  title={On intervals, transitivity= chaos},
  author={Vellekoop, Michel and Berglund, Raoul},
  journal={The American Mathematical Monthly},
  volume={101},
  number={4},
  pages={353--355},
  year={1994}
}

@article{Zgliczynski96:Fixed,
  title={Fixed point index for iterations of maps, topological horseshoe and chaos},
  author={Zgliczy{\'n}ski, Piotr},
  journal={Topological Methods in Nonlinear Analysis},
  volume={8},
  number={1},
  pages={169--177},
  year={1996}
}

@article{Adler65:Topological,
  title={Topological entropy},
  author={Adler, Roy L. and Konheim, Alan G. and McAndrew, M. H.},
  journal={Transactions of the American Mathematical Society},
  volume={114},
  number={2},
  pages={309--319},
  year={1965},
  publisher={American Mathematical Society}
}

@article{Bowen71:Periodic,
  title={Periodic points and measures for Axiom A diffeomorphisms},
  author={Bowen, Rufus},
  journal={Transactions of the American Mathematical Society},
  volume={154},
  pages={377--397},
  year={1971}
}

@article{Myerson81:Optimal,
  author    = {Myerson, Roger B.},
  title     = {Optimal Auction Design},
  journal   = {Mathematics of Operations Research},
  volume    = {6},
  number    = {1},
  pages     = {58--73},
  year      = {1981}
}

@article{Chua86:Double,
  title = {The double scroll family},
  author = {Chua, Leon O and Komuro, Motomasa and Matsumoto, Takashi},
  journal = {IEEE Transactions on Circuits and Systems},
  volume = {33},
  number = {11},
  pages = {1072--1118},
  year = {1986},
  publisher = {IEEE}
}

@article{Ausubel06:Lovely,
  title={The lovely but lonely {V}ickrey auction},
  author={Ausubel, Lawrence M and Milgrom, Paul},
  journal={Combinatorial auctions},
  volume={17},
  number={3},
  pages={22--26},
  year={2006}
}

@article{Feigenbaum78:Quantitative,
  title={Quantitative universality for a class of nonlinear transformations},
  author={Feigenbaum, Mitchell J},
  journal={Journal of Statistical Physics},
  volume={19},
  number={1},
  pages={25--52},
  year={1978}
}

@book{NemirovskiYudin83:Problem,
  author    = {Nemirovski, Arkadii Semenovich and Yudin, David Borisovich},
  title     = {Problem Complexity and Method Efficiency in Optimization},
  year      = {1983},
  publisher = {John Wiley \& Sons}
}

@article{Hammel88:Numerical,
author = {Stephan M. Hammel and James A. Yorke and Celso Grebogi},
title = {{Numerical orbits of chaotic processes represent true orbits}},
volume = {19},
journal = {Bulletin (New Series) of the American Mathematical Society},
number = {2},
publisher = {American Mathematical Society},
pages = {465 -- 469},
year = {1988}
}

@article{Banks92:Devaney,
  title={On Devaney's definition of chaos},
  author={Banks, John and Brooks, Jeffrey and Cairns, Grant and Davis, Gary and Stacey, Peter},
  journal={The American mathematical monthly},
  volume={99},
  number={4},
  pages={332--334},
  year={1992}
}

@book{Devaney89:Introduction,
  title={An Introduction to Chaotic Dynamical Systems},
  author={Devaney, Robert L.},
  year={1989},
  publisher={Addison-Wesley},
  address={Reading, MA},
  edition={2nd}
}

@book{Elaydi07:Discrete,
  title={Discrete Chaos: With Applications in Biology and Engineering},
  author={Elaydi, Saber N.},
  year={2007},
  publisher={Chapman and Hall/CRC},
  edition={2nd}
}

\clearpage
\appendix

\section{Further background}
\label{sec:moreprels}

This section provides further background on discrete-time dynamical systems, introducing some key concepts upon which~\Cref{sec:discrete} relies.

\subsection{Li-Yorke chaos}

Let $m^{(t+1)} = F(m^{(t)})$ be a one-dimensional, discrete-time dynamical system, where $F$ maps an interval $J$ to itself. We will say that a point $m \in J$ is \emph{periodic with period $p$} if $m = F^{(p)}(m)$ and $m \neq F^{(t)}(m)$ for all $t < p$. We say that $m$ is periodic if it is periodic for some period $p \geq 1$. We are now ready to state the seminal result of~\citet{Li75:Period}.

\begin{theorem}[\citealp{Li75:Period}]
    \label{theorem:LiYorke}
    Let $J$ be an interval and $F : J \to J$ be continuous. Suppose further there is a point $a \in J$ such that the points $b \defeq F(a)$, $c \defeq F(b)$, and $d \defeq F(c)$ satisfy $d \leq a < b < c$ or $d \geq a > b > c$. Then for every $t = 1, 2, \dots,$ there is a periodic point in $J$ with period $t$. Furthermore, there is an uncountable set $S \subseteq J$ (containing no periodic points) that satisfies the following conditions.
    \begin{enumerate}
        \item For every $p, q \in S$ with $p \neq q$,\label{item:scrambled}
        \begin{equation*}
            \limsup_{t \to \infty} | F^{(t)}(p) - F^{(t)}(q) | > 0 \text{ and } \liminf_{t \to \infty} | F^{(t)}(p) - F^{(t)}(q) | = 0.
        \end{equation*}
        \item For every $p \in S$ and periodic point $q \in J$,\label{item:periodic}
        \begin{equation*}
            \limsup_{t \to \infty} | F^{(t)}(p) - F^{(t)}(q) | > 0.
        \end{equation*}
    \end{enumerate}
\end{theorem}

The implication given by~\Cref{item:scrambled} is referred to as $S$ being \emph{scrambled}; it is worth noting that some authors make the stronger assumption $\limsup_{t \to \infty} | F^{(t)}(p) - F^{(t)}(q) | = 1$ (when $J \defeq [0, 1]$)~\citep{Bruckner87:Scrambled}. A system that satisfies the conclusion of~\Cref{theorem:LiYorke} is said to exhibit \emph{Li-Yorke chaos}. We refer to~\citet{Sharkovskii07:Co} for a precursor of~\Cref{theorem:LiYorke} related to period-ordering constraints.

Li-Yorke chaos is one possible manifestation of chaos. One limitation is that the scrambled set may possess (Lebesgue) measure zero---a complex set may exist, but may not be observed under a random initialization. We proceed with other notions of chaos.

\subsection{Sensitivity to initial conditions}

One of the hallmarks of a chaotic system is its sensitivity to initial conditions, popularly known as the ``butterfly effect.'' This means that a small error in the initialization will be amplified over time. The important practical consequence is that, in such a system, computer simulations can be misleading~\citep{Elaydi07:Discrete}.

\begin{definition}[Sensitivity to initial conditions]
    \label{def:sensitivity}
    A map $F$ of an interval $I$ is said to possess \emph{sensitive dependence on initial conditions} if there exists $\nu > 0$ such that for any $x_0 \in I$ and $\delta > 0$, there exists $x_0' \in I \cap (x_0 - \delta, x_0 + \delta)$ and a positive integer $k$ such that
    \begin{equation*}
        \left| F^{(k)}(x_0) - F^{(k)}(x_0') \right| \geq \nu.
    \end{equation*}
    The number $\nu$ is called the \emph{sensitivity constant} of $F$.
\end{definition}

The simplest example of sensitivity to initial conditions is a linear map $x \mapsto c x$, for $c > 1$. However, this is not an interesting example as it does not possess the other properties necessary for chaos. Quantitatively, sensitivity per~\Cref{def:sensitivity} is characterized by a \emph{positive Lyapunov exponent}, which measures the average exponential rate of divergence of nearby orbits. The largest Lyapunov exponent can be numerically estimated (\Cref{appendix:LLE}); a positive (estimated) Lyapunov exponent for a bounded system is a common working definition of chaos.

\subsection{Transitivity}

We next state the definition of \emph{topological transitivity}~\citep{Elaydi07:Discrete}.

\begin{definition}[Topological transitivity]
    \label{def:transitivity}
    Let $F$ be a map on an interval $I$ (or $\mathbb{R}$). Then $F$ is said to be \emph{topologically transitive} if for any pair of nonempty open intervals $J_1$ and $J_2$ in $I$, there exists a positive integer $k$ such that
    \begin{equation*}
        F^{(k)}(J_1) \cap J_2 \neq \emptyset.
    \end{equation*}
\end{definition}
Equivalently, one may replace the intervals $J_1$ and $J_2$ with open subsets $U_1$ and $U_2$ of $I$. (An open set here is just the union of open intervals.)

A convenient criterion for proving topological transitivity for a map $F: I \to I$ on the interval $I$ is that it has a dense orbit~\citep{Elaydi07:Discrete}.

\subsection{Devaney chaos}

Armed with the previous definitions, we next state the definition of chaos in the classic sense of~\citet{Devaney89:Introduction}.

\begin{definition}[Devaney chaos]
    \label{def:Devaney}
    A map $F : I \to I$, where $I$ is an interval, is said to be \emph{chaotic} if
    \begin{enumerate}
        \item $F$ is topologically transitive (\Cref{def:transitivity});
        \item The set of periodic points $P$ is dense in $I$; and
        \item $F$ has sensitive dependence on initial conditions (\Cref{def:sensitivity}).
    \end{enumerate}
\end{definition}

\begin{remark}
It was subsequently found by~\citet{Banks92:Devaney} that the first two conditions imply the third. Moreover, \citet{Vellekoop94:Intervals} showed that for continuous maps on intervals in $\R$, transitivity implies that the set of periodic points is dense. In other words, topological transitivity implies Devaney chaos. Nonetheless, it is still common to define Devaney chaos as in~\Cref{def:Devaney}.
\end{remark}

For a continuous map $F : I \to I$ on a closed and bounded interval $I$, Devaney's chaos is equivalent to having \emph{positive topological entropy}~\citep{Elaydi07:Discrete}. In the one-dimensional setting, the topological entropy, which we denote by $h$, is a measure of the growth of the number of periodic cycles as a function of the length of the period; namely,
\begin{equation*}
    h(F) = \lim_{t \to \infty} \frac{\ln N_t}{t},
\end{equation*}
where $N_t$ is the number of distinct periodic orbits of length $t$. The more general high-dimensional setting is treated later in~\Cref{sec:top-entropy}. To put the foregoing definitions into context, it is worth noting that there are examples in which Devaney chaos does not imply Li-Yorke chaos~\citep[Example 3.15]{Elaydi07:Discrete}.

To prove that a map $F$ is topologically transitive, one can establish a \emph{conjugacy}\footnote{Let $f: X \to X$ and $g: Y \to Y$ be two continuous maps on topological spaces $X$ and $Y$. $f$ and $g$ are said to be \emph{conjugate} if there exists a homeomorphism $h: X \to Y$ (a continuous bijection with a continuous inverse) such that $h \circ f = g \circ h$.} with the shift map~\citep{Elaydi07:Discrete}; this can be shown for the Ricker map with respect to a certain set $\Lambda$, leading to~\Cref{thm:Devaney} claimed earlier.

\subsection{Unimodal maps}

Finally, to put~\Cref{sec:symmetry} in a better context, we state the definition of an $S$-unimodal map. The canonical example of such a map is the logistic map $x \mapsto r x (1 - x)$.

\begin{definition}[S-unimodal map; \citealp{Singer78:Stable}]
    \label{def:unimodal}
    A $C^3$ map $F : [a, b] \to [a, b]$ is \emph{$S$-unimodal} if
    \begin{enumerate}
        \item either $F(a) = F(b) = a$ or $F(a) = F(b) = b$;
        \item $F$ has a unique critical point $c \in (a, b)$;
        \item the Schwarzian derivative of $F$, defined as
        \begin{equation*}
            SF(x) = \frac{F'''(x)}{F'(x)} - \frac{3}{2} \left( \frac{F''(x)}{F'(x)} \right)^2,
        \end{equation*}
        is negative for all $x \in [a, b]$ such that $x \neq c$.
    \end{enumerate}
\end{definition}
\section{Framework for proving topological chaos}
\label{app:galias_framework}

In this section, we outline the theoretical framework employed by~\citet{Galias97:Positive} to rigorously establish that Chua’s circuit has positive topological entropy. We earlier used the framework of~\citet{Galias97:Positive} to confirm that our approximate version of Chua's circuit---which relies on continuous negation---behaves similarly to the original system (\Cref{fig:quadrangles-approx}).

\subsection{Fixed point index and homotopy property}

The argument of~\citet{Galias97:Positive} makes use of \emph{fixed point index theory}~\citep{Dold12:Lectures}. In what follows, $\partial U$ denotes the boundary of $U$; $\widebar{U}$ the closure of $U$; and $\interior U$ the interior of $U$. Let $U$ be an open and bounded set. For a continuous map $F: V \to \mathbb{R}^n$, such that $\widebar{U} \subset V \subset \R^d$ and with no fixed points on $\partial U$, the fixed point index $I(F, U)$ is an integer that provides a topological count of the fixed points; the precise definition is not important for our purposes here. A basic consequence of the definition is that if $I(F, U) \neq 0$, then $F$ admits at least one fixed point in $U$. 

A useful technique in this theory is the \textit{homotopy method}. Specifically, if two maps $F$ and $F'$ can be continuously transformed into each other through a \emph{homotopy} $h(t)$, where $t \in [0,1]$, such that no fixed points of $h(t)$ lie on the boundary $\partial U$ for any $t$, then their fixed point indices are equal: $I(F, U) = I(F', U)$. This allows the characterization of a map by relating it to a known template map.

\subsection{The Smale horseshoe map and its deformation}

Let $P = [-1, 1] \times \mathbb{R}$ be a vertical stripe in the plane $\mathbb{R}^2$. Within this stripe, we consider two disjoint compact sets $N_0$ and $N_1$ defined as
\begin{equation*}
    N_0 = [-1, 1] \times [-1, -0.5] \text{ and } N_1 = [-1, 1] \times [0.5, 1].
\end{equation*}
The regions of $P$ relative to these sets are $M_{-}$ (below $N_0$), $M_{0}$ (between $N_0$ and $N_1$), and $M_{+}$ (above $N_1$). Also, let $N_{0 D}, N_{0 U}$ be the lower and upper horizontal edges of $N_0$ and $N_{0 L}, N_{0 R}$ be the left and right vertical edges of $N_0$; similar notation is used with respect to $N_1$.

The standard \emph{Smale horseshoe}, denoted by $h_s$, is a linear map on $N_0$ and $N_1$, defined as
\begin{equation*}
h_s(x,y) \defeq \begin{cases} \left( \frac{1}{4} x - \frac{1}{2}, \pm 5 \left(y + \frac{3}{4} \right) \right) & \text{for } (x,y) \in N_0, \\ \left(\frac{1}{4} x + \frac{1}{2} , \pm 5 \left(y - \frac{3}{4} \right) \right) & \text{for } (x,y) \in N_1. \end{cases}
\end{equation*}
The coefficient above determines the orientation (\Cref{fig:horseshoe}). It can be shown that $h_s^{(t)}$ has $2^t$ different fixed points~\citep{Galias97:Positive}.

Galias considered the \emph{deformed horseshoe map} $h_d$. It remains linear on both rectangles, but the stretching action on $N_1$ is weakened:
\begin{equation*}
h_d(x,y) = \begin{cases} \left( \frac{1}{4} x - \frac{1}{2}, \pm 5(y + \frac{3}{4} ) \right) & \text{for } (x,y) \in N_0, \\ \left(\frac{1}{4} x + \frac{1}{2}, \pm 2 \left(y - \frac{3}{4} \right) - \frac{3}{4} \right) & \text{for } (x,y) \in N_1. \end{cases}
\end{equation*}
The critical distinction is that $h_d(N_1) \cap N_1 = \emptyset$, meaning that the transition $1 \to 1$ is forbidden, as shown in~\Cref{fig:horseshoe}.

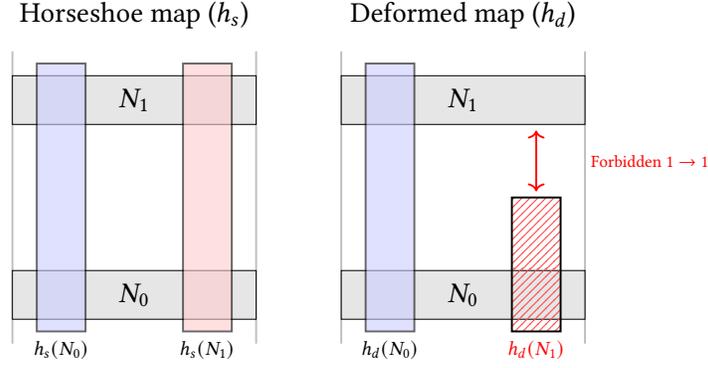
\begin{figure}[h]
    \centering
    \scalebox{0.9}{    \begin{tikzpicture}[scale=1.8]
        \node at (0, 1.5) {Horseshoe map ($h_s$)};
        
        \draw[thick, gray!50] (-1, -1.2) -- (-1, 1.2);
        \draw[thick, gray!50] (1, -1.2) -- (1, 1.2);
        
        \draw[fill=gray!20] (-1, -1) rectangle (1, -0.6);
        \node at (0, -0.8) {$N_0$};
        \draw[fill=gray!20] (-1, 0.6) rectangle (1, 1);
        \node at (0, 0.8) {$N_1$};
        
        \draw[thick, fill=blue!20, opacity=0.6] (-0.8, -1.1) rectangle (-0.4, 1.1);
        \node[scale=0.7] at (-0.6, -1.25) {$h_s(N_0)$};
        
        \draw[thick, fill=red!20, opacity=0.6] (0.4, -1.1) rectangle (0.8, 1.1);
        \node[scale=0.7] at (0.6, -1.25) {$h_s(N_1)$};
    \end{tikzpicture}
    \hspace{1cm}
    \begin{tikzpicture}[scale=1.8]
        \node at (0, 1.5) {Deformed map ($h_d$)};
        
        \draw[thick, gray!50] (-1, -1.2) -- (-1, 1.2);
        \draw[thick, gray!50] (1, -1.2) -- (1, 1.2);
        
        \draw[fill=gray!20] (-1, -1) rectangle (1, -0.6);
        \node at (0, -0.8) {$N_0$};
        \draw[fill=gray!20] (-1, 0.6) rectangle (1, 1);
        \node at (0, 0.8) {$N_1$};
        
        \draw[thick, fill=blue!20, opacity=0.6] (-0.8, -1.1) rectangle (-0.4, 1.1);
        \node[scale=0.7] at (-0.6, -1.25) {$h_d(N_0)$};
        
        \draw[thick, pattern=north east lines, pattern color=red!70] (0.4, -1.1) rectangle (0.8, 0);
        \node[scale=0.7, red] at (0.6, -1.25) {$h_d(N_1)$};
        
        \draw[<->, red, thick] (0.6, 0.05) -- (0.6, 0.55);
        \node[red, right, scale=0.6] at (1.0, 0.3) {Forbidden $1 \to 1$};
    \end{tikzpicture}}
    \caption{Standard and deformed horseshoe maps.}
    \label{fig:horseshoe}
\end{figure}

\subsection{Existence of infinitely many periodic orbits}
The following results formalize the conditions under which a continuous map admits infinitely many periodic orbits by way of a homotopy with the (deformed) horseshoe map.

\begin{theorem}[Existence for horseshoe maps; \citealp{Zgliczynski96:Fixed}]
Suppose that $F(N_0), F(N_1) \subset \interior P$, the horizontal edges of $N_0, N_1$ are mapped by $F$ outside of $N_0 \cup N_1$ in such a way that one of the sets $F(N_{0D}), F(N_{0U})$ is enclosed in $M_+$ while the second one is enclosed in $M_-$, and similarly for horizontal edges of $N_1$. Then for any finite sequence $a_0, a_1, \dots, a_{t-1} \in \{0, 1\}^t$, there exists a point $x$ satisfying $F^{(\tau)}(x) \in N_{a_i}$ for $\tau = 0, \dots, t-1$ and $F^{(t)}(x) = x$.
\end{theorem}

In what follows, $T_t \defeq \{ (a_0, \dots, a_{t-1}) \in \{0, 1\}^t : (a_{\tau}, a_{\tau + 1 \mod t }) \neq (1, 1) \text{ for } 0 \leq \tau < t \}$. It is easy to show that $|T_t| = (\frac{1 + \sqrt{5}}{2})^t + (\frac{1 - \sqrt{5}}{2})^t$. It can also be shown that $|T_t|$ is the number of fixed points of $h_d^{(t)}$~\citep{Galias97:Positive}.

\begin{theorem}[Existence for deformed horseshoes; \citealp{Galias97:Positive}]
\label{theorem:deformed}
Suppose that $F(N_0), F(N_1) \subset \operatorname{int} P$, and the horizontal edges of $N_0, N_1$ are mapped by $F$ in such a way that one of the sets $F(N_{0D})$, $F(N_{0U})$ is enclosed in $M_+$ while the second one is enclosed in $M_-$, and one of the sets $F(N_{1D})$, $F(N_{1U})$ is enclosed in $M_-$ while the second one is enclosed in $M_0 \cup N_0 \cup M_+$. Then for any finite sequence $a = (a_0, a_1, \dots, a_{t-1}) \in T_t$, there exists a point $x$ satisfying $F^{(\tau)}(x) \in N_{a_i}$ for $\tau = 0, \dots, t-1$ and $F^{(t)}(x) = x$.
\end{theorem}

One can apply this theorem more generally by placing $N_0$ and $N_1$ on different positions on the real plane, but suitably changing the coordinate system (\Cref{fig:quadrangles}). \Cref{theorem:deformed} tells us that proving the existence of infinitely many periodic orbits reduces to characterizing the images of $N_0$ and $N_1$ under $F$; this is where a computer-assisted proof comes handy. The following lemma is useful for reducing the amount of computation needed.

\begin{lemma}
If $F$ is a one-to-one map and the images of the boundaries $\partial N_0, \partial N_1$ are contained in $P$, then $F(N_0) \subset P$ and $F(N_1) \subset P$.
\end{lemma}

\subsection{Poincar\'e map}
\label{sec:Poincare}

To characterize the three-dimensional flow arising from Chua's circuit, as defined in~\eqref{eq:Chua}, Galias defines a suitable two-dimensional discrete map referred to as the \emph{Poincar\'e map}. (More broadly, this is a general method to reduce problems in continuous time to corresponding problems in discrete time.) The \emph{traversal plane} is defined as $\Sigma = \{ (x, y, z) \in \R^3 : x = 1 \} $. The Poincar\'e map is defined by $P_\Sigma(\vec{x}) = \phi_{\tau(\vec{x})}(\vec{x})$, where $\phi_t(\vx)$ is the trajectory of Chua's circuit~\eqref{eq:Chua} starting at $\vx$, and $\tau(\vx)$ is the time needed to return to $\Sigma$. Galias defines $N_0$ and $N_1$ in such a way that the Poincar\'e map restricted to those sets is continuous. 

In particular, this is done as follows. On the transversal plane $\Sigma$, the following 8 points are chosen.
\begin{align*}
    A_1 &= (-0.1950, -2.6942956550), & A_5 &= (-0.3181, -4.1785885539), \\
    A_2 &= (-0.1761, -2.2243882059), & A_6 &= (-0.3315, -4.0981421985), \\
    A_3 &= (-0.2376, -2.9659317744), & A_7 &= (-0.3597, -4.4381670543), \\
    A_4 &= (-0.2410, -3.2489461290), & A_8 &= (-0.3472, -4.5294652668).
\end{align*}
These points lie on two parallel lines:
\begin{equation*}
    z = 9.623 (1.253 y - 0.0105) \quad \text{and} \quad z = 9.623 (1.253 y - 0.03565).
\end{equation*}
Let $N_1$ be the quadrangle $A_1 A_2 A_3 A_4$, and $N_0$ be the quadrangle $A_5 A_6 A_7 A_8$. Further, we let $N_{1U} = A_1 A_2, N_{1D} = A_3 A_4, N_{0U} = A_5 A_6, N_{0D} = A_7 A_8$. An illustration is given in~\Cref{fig:quadrangles}.

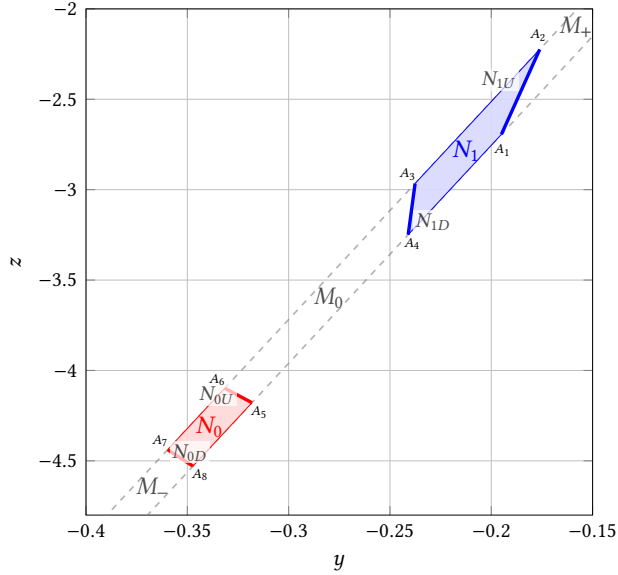
\begin{figure}
    \centering
    \scalebox{0.8}{\begin{tikzpicture}
\begin{axis}[
    title={},
    xlabel={$y$},
    ylabel={$z$},
    grid=major,
    xmin=-0.40, xmax=-0.15,
    ymin=-4.8, ymax=-2.0,
    width=10cm, height=10cm,
    label style={font=\small},
    tick label style={font=\footnotesize},
    legend style={at={(0.98,0.02)}, anchor=south east, font=\footnotesize, fill=white, fill opacity=0.8}
]

    \coordinate (A1) at (-0.1950, -2.6942956550);
    \coordinate (A2) at (-0.1761, -2.2243882059);
    \coordinate (A3) at (-0.2376, -2.9659317744);
    \coordinate (A4) at (-0.2410, -3.2489461290);
    \coordinate (A5) at (-0.3181, -4.1785885539);
    \coordinate (A6) at (-0.3315, -4.0981421985);
    \coordinate (A7) at (-0.3597, -4.4381670543);
    \coordinate (A8) at (-0.3472, -4.5294652668);

    
    \addplot[
        domain=-0.40:-0.15,
        samples=2, 
        color=gray,
        thick,
        dashed,
        opacity=0.6
    ] {(1.253*x - 0.0105) * 9.623} coordinate[pos=0.85] (L1_label);

    \addplot[
        domain=-0.40:-0.15,
        samples=2,
        color=gray,
        thick,
        dashed,
        opacity=0.6
    ] {(1.253*x - 0.03565) * 9.623} coordinate[pos=0.9] (L2_label);

    \node[gray, anchor=west, font=\footnotesize] at (L1_label) {};
    \node[gray, anchor=west, font=\footnotesize] at (L2_label) {};

    \fill[blue!15, opacity=0.9] (A1) -- (A2) -- (A3) -- (A4) -- cycle;
    \draw[blue, thin] (A2) -- (A3);
    \draw[blue, thin] (A4) -- (A1);
    \draw[blue, ultra thick] (A1) -- (A2) node[midway, above left, font=\footnotesize, text=black, fill=white, fill opacity=0.7, inner sep=1pt] {$N_{1U}$};
    \draw[blue, ultra thick] (A3) -- (A4) node[midway, below right, font=\footnotesize, text=black, fill=white, fill opacity=0.7, inner sep=1pt] {$N_{1D}$};
    \node[blue, font=\normalsize] at (barycentric cs:A1=1,A2=1,A3=1,A4=1) {$N_1$};

    \fill[red!15, opacity=0.9] (A5) -- (A6) -- (A7) -- (A8) -- cycle;
    \draw[red, thin] (A6) -- (A7);
    \draw[red, thin] (A8) -- (A5);
    \draw[red, ultra thick] (A5) -- (A6) node[midway, above left, font=\footnotesize, text=black, fill=white, fill opacity=0.7, inner sep=1pt,yshift=-0.2cm] {$N_{0U}$};
    \draw[red, ultra thick] (A7) -- (A8) node[midway, below right, font=\footnotesize, text=black, fill=white, fill opacity=0.7, yshift=0.3cm, xshift=-0.2cm, inner sep=1pt] {$N_{0D}$};
    \node[red, font=\normalsize] at (barycentric cs:A5=1,A6=1,A7=1,A8=1) {$N_0$};

    \node[above left, xshift=0.3cm, yshift=-0.5cm, font=\tiny] at (A1) {$A_1$};
    \node[above right, xshift=-0.3cm, yshift=0, font=\tiny] at (A2) {$A_2$};
    \node[below right, xshift=-0.4cm, yshift=0.4cm, font=\tiny] at (A3) {$A_3$};
    \node[below left, xshift=0.35cm, yshift=0.1cm, font=\tiny] at (A4) {$A_4$};

    \node[above left, xshift=0.4cm, yshift=-0.4cm, font=\tiny] at (A5) {$A_5$};
    \node[above right, xshift=-0.4cm, yshift=-0.1cm, font=\tiny] at (A6) {$A_6$};
    \node[below right, yshift=0.4cm, xshift=-0.4cm, font=\tiny] at (A7) {$A_7$};
    \node[below left, xshift=0.4cm, yshift=0.1cm, font=\tiny] at (A8) {$A_8$};

    
    \node[font=\normalsize, anchor=center, black!70] at (axis cs: -0.158, -2.1) {$M_+$};

    \node[font=\normalsize, anchor=center, black!70] at (axis cs: -0.28, -3.6) {$M_0$};

    \node[font=\normalsize, anchor=center, black!70] at (axis cs: -0.367, -4.65) {$M_-$};

\end{axis}
\end{tikzpicture}}
    \caption{The definition of $N_0$ and $N_1$ upon which the Poincar\'e map is defined.}
    \label{fig:quadrangles}
\end{figure}

\subsection{Topological entropy}
\label{sec:top-entropy}

We next recall the notion of topological entropy for discrete-time systems~\citep{Adler65:Topological}. We will introduce an equivalent definition based on the notion of \emph{separated sets}.

\begin{definition}[Separated sets]
    Let $(X, \rho)$ be a compact metric space. A set $E \subset X$ is called $(t, \epsilon)$-separated if for every two distinct points $x, x' \in E$, there exists $\tau \in \{0, \dots, t-1 \}$ such that $\rho( F^{(\tau)}(x), F^{(\tau)}(x') ) > \epsilon$.
\end{definition}

\begin{theorem}[\citealp{Bowen71:Periodic}]
    \label{theorem:top-entropy}
    The topological entropy of a map $F$ can be expressed as 
    \begin{equation*}
        h(F) = \lim_{\epsilon \to 0} \limsup_{t \to \infty} \frac{1}{t} \log s_t(\epsilon),
    \end{equation*}
    where $s_t(\epsilon) \defeq \max \{\card E : E \text{ is } (t, \epsilon)\text{-separated} \} $.
\end{theorem}

For the Poincar\'e map $P_\Sigma$ defined in~\Cref{sec:Poincare}, which is induced by Chua's circuit~\eqref{eq:Chua}, \citet{Galias97:Positive} used~\Cref{theorem:top-entropy} to prove that $h(P_\Sigma) \geq \log \frac{1 + \sqrt{5}}{2}$. In particular, his argument lower bounds $s_{t}(\epsilon)$ by $|T_t|$, which is in turn at least $( \frac{1 + \sqrt{5}}{2} )^t$ ; the lower bound on the topological entropy of $P_\Sigma$ then follows from~\Cref{theorem:top-entropy}. This part of the argument directly extends to our setting.

The final step is to use the fact that the Poincar\'e map has positive topological entropy to conclude the same for the flow corresponding to Chua's circuit. 

\begin{definition}
    The topological entropy of a flow $(X, F)$ is the topological entropy of the map $\pi_1$, where $\pi_t : X \to X$ is defined as $\pi_t(\vx) = F(\vx, t)$.
\end{definition}

The existence of infinitely many periodic orbits of $P_\Sigma$ on $N_0 \cup N_1$ can be used to prove that the topological entropy of the flow is positive. The argument given by~Galias extends directly in our setting.
\section{Omitted details}
\label{sec:details}

This section contains additional details and background on our numerical simulations. 

\subsection{Cobweb plot}
\label{sec:cobweb}

In~\Cref{fig:cobweb}, we presented several \emph{cobweb plots}~\citep{Strogatz94:Nonlinear} to visualize the trajectory of the system for different values of the learning rate. Cobwebs are a standard tool in dynamical systems because they provide immediate geometric intuition regarding the behavior of the orbit. They work as follows. Starting from an initial state $m^{(0)}$ on the horizontal axis, we draw a vertical line to the update map to find the next value $m^{(1)}$. We then project this value horizontally on the identity line. This values serves as the input for the next iteration. Repeating this vertical-to-curve and horizontal-to-identity pattern traces the full orbit of the dynamics.

\subsection{Bifurcation diagram}
\label{sec:bifurcation}

A mainstay tool to capture the global behavior of a system as a function of its parameter---in our case the learning rate---is the \emph{bifurcation diagram}. The horizontal axis represents the control parameter $\eta$, while the vertical axis shows the set of values asymptotically attained by the system; the initial values are not included to discount transient phenomena. For each $\eta$, the system is simulated for a sufficiently long period; the iterates are then plotted vertically. As a result, a curve comprising a single point shows convergence to a fixed point, while a branching into $k$ distinct points reveals a period-$k$ limit cycle. The disintegration of these branches into dense, scattered regions marks the onset of chaos.

\Cref{fig:bifurcation} shows the bifurcation diagram of the dynamical system~\eqref{eq:expmap}, corresponding to entropic mirror descent, whereas~\Cref{fig:GD-bifurcation} shows the one for~\eqref{eq:GD}---Euclidean mirror descent.

\begin{figure}[t]
    \centering
    \includegraphics[scale=0.35]{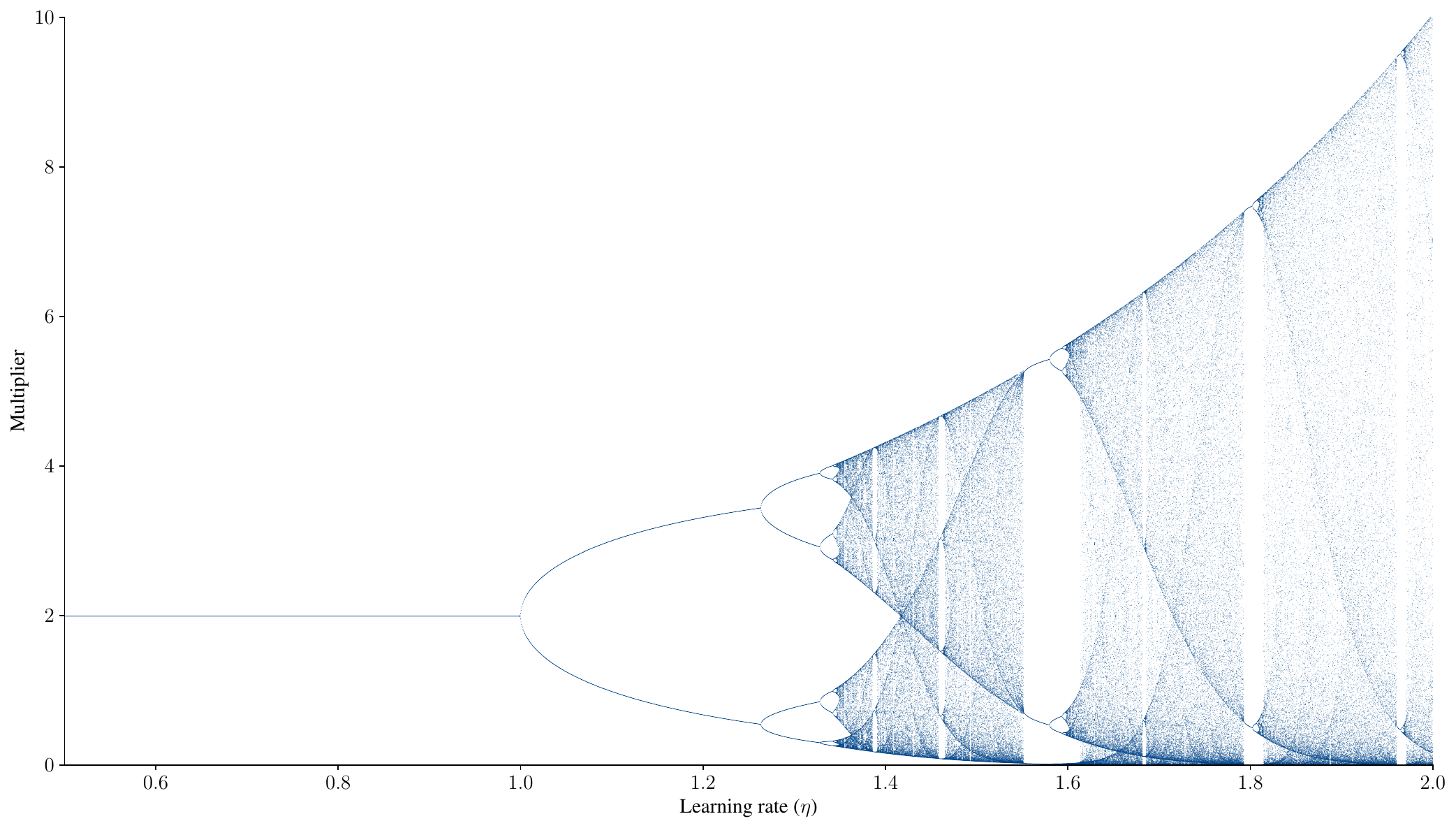}
    \caption{Bifurcation diagram for the system~\eqref{eq:expmap} under the valuation profile given in~\eqref{eq:param-valuations} for $v = 2$.}
    \label{fig:bifurcation}
\end{figure}

\begin{figure}[h]
    \centering
    \includegraphics[scale=0.35]{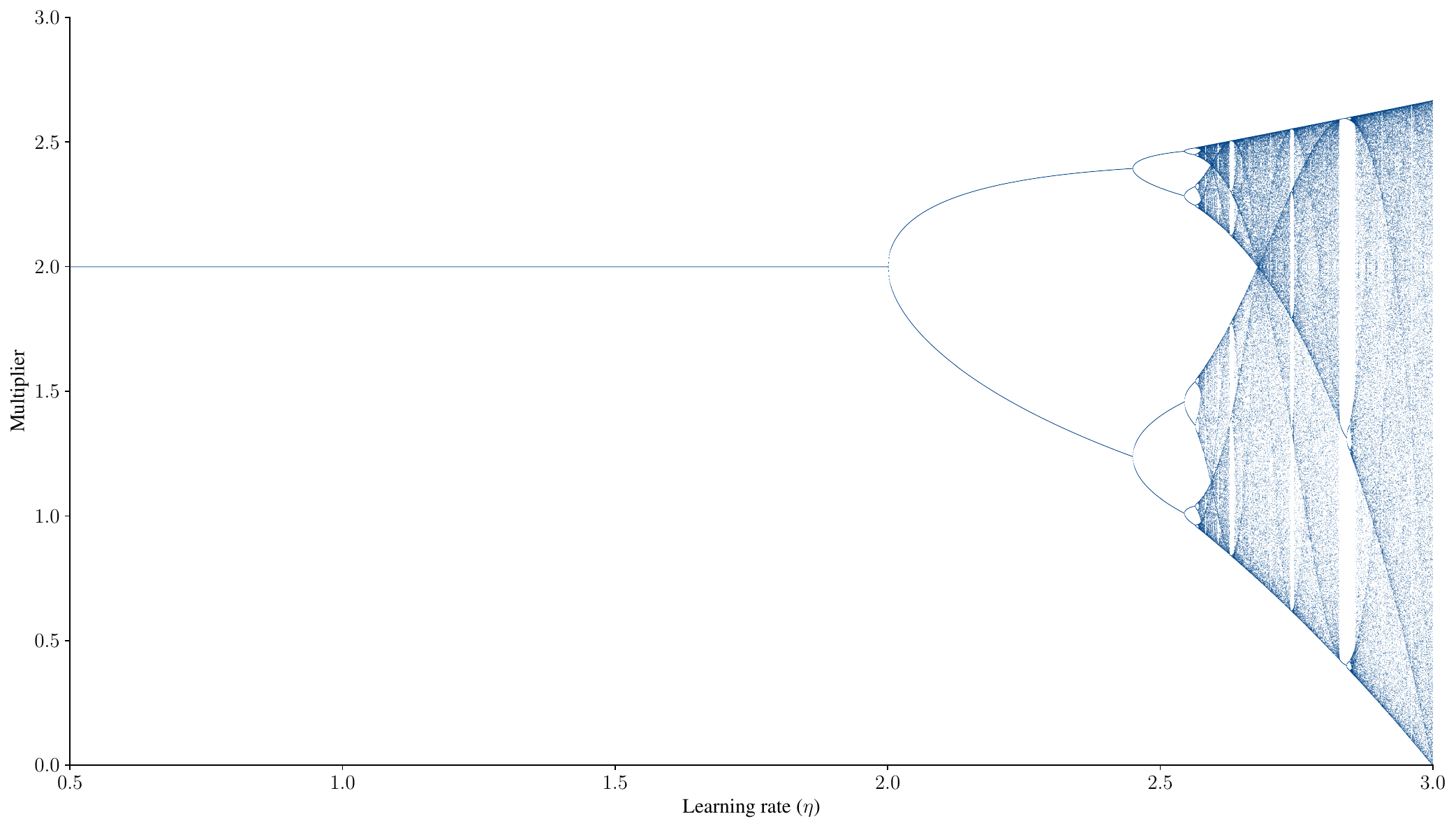}
    \caption{Bifurcation diagram for the system~\eqref{eq:GD} under the valuation profile given in~\eqref{eq:param-valuations} for $v = 2$.}
    \label{fig:GD-bifurcation}
\end{figure}

\subsection{Lyapunov exponents}
\label{appendix:LLE}

A standard way to numerically confirm the onset of chaos is by estimating the largest \emph{Lyapunov exponent}. This serves as a measure of the average rate at which two infinitesimally close trajectories separate. A positive Lyapunov exponent is important as it implies sensitivity to initial conditions: microscopic uncertainties are amplified exponentially over time.

Our numerical estimation follows the standard Benettin algorithm. We simulate a reference trajectory alongside a shadow trajectory perturbed by a tiny magnitude $d_0 = 10^{-8}$. Because chaotic orbits diverge rapidly, we do not let the shadow trajectory run free, as it would quickly separate too far to capture the local rate of expansion. Instead, at fixed time intervals of $\Delta t = 0.5$, we measure the new Euclidean distance $d(t)$, record the logarithmic expansion rate $\ln(d(t)/d_0)$, and then pull back the shadow trajectory toward the reference to restore the initial separation $d_0$ while preserving its direction. The largest Lyapunov exponent is computed as the long-term time average of these local expansion rates. For the experiment reported in the main body, we integrated over a total time horizon of $T=2000$, discarding the first 200 time units to eliminate transient effects.

\subsection{Chua's circuit}
\label{appendix:Chuacircuit}

We continue by providing further details on Chua's circuit. For the convenience of the reader, let us first recall the state equations governing its evolution.

\begin{figure}[!h]
    \centering
    \scalebox{0.8}{\begin{tikzpicture}
    \pgfmathsetmacro{\Ga}{-3.4429} 
    \pgfmathsetmacro{\Gb}{-2.1849} 
    \pgfmathsetmacro{\E}{1}     
    \pgfmathsetmacro{\Ya}{\Ga*\E} 

    \begin{axis}[
        axis lines = center,
        xlabel = {$x$ (Voltage)},
        ylabel = {$g(x)$ (Current)},
        xmin = -2.5, xmax = 2.5,
        ymin = -7, ymax = 7,
        grid = major,
        grid style = {dashed, gray!30},
        width = 10cm,
        height = 10cm, 
        xtick = {-\E, \E},
        xticklabels = {$-1$, $1$},
        ytick = {-\Ya, \Ya},
        yticklabels = {},
        axis line style = {-stealth, thick},
        ticklabel style={font=\small},
        xlabel style={at={(ticklabel* cs:1)},anchor=north west},
        ylabel style={at={(ticklabel* cs:1)},anchor=south east}
    ]

    \addplot[
        mark=none,
        domain=-2:2,
        color=red!80!black,
        ultra thick,
        samples at={-2, -\E, \E, 2} 
    ]
    coordinates{
        (-2,  {-\Ya + \Gb*(-2 - (-\E))}) 
        (-\E, {-\Ya})                   
        (\E,  {\Ya})                    
        (2,   {\Ya + \Gb*(2 - \E)})     
    };

    
    \draw[dashed, gray] (axis cs:\E, 0) -- (axis cs:\E, \Ya) -- (axis cs:0, \Ya);
    \draw[dashed, gray] (axis cs:-\E, 0) -- (axis cs:-\E, -\Ya) -- (axis cs:0, -\Ya);
    
    
    \node[anchor=west, font=\footnotesize] at (axis cs: 0.47, -1.5) {$G_a \approx -3.44$};
    
    \node[anchor=west, font=\footnotesize] at (axis cs: 1, -6) {$G_b \approx -2.18$};
    \node[anchor=east, font=\footnotesize] at (axis cs: -1.2, 6) {$G_b \approx -2.18$};
    \end{axis}
\end{tikzpicture}}
    \caption{The characteristic curve $g(x)$ of Chua's diode, defined in~\eqref{eq:nonlinear-Chua}.}
    \label{fig:diode}
\end{figure}
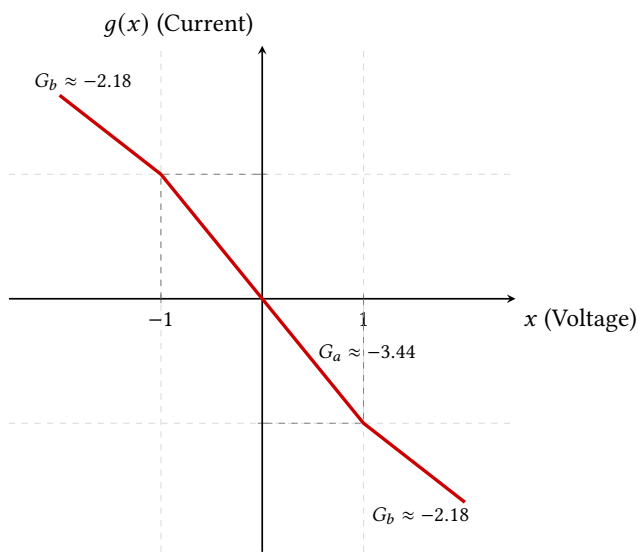

\begin{equation*}\begin{gathered}C_1 \frac{dx}{dt} = G(y - x) - g(x),\\ C_2 \frac{dy}{dt} = G(x - y) + z, \\
L \frac{dz}{dt} = -y - R_0 z.
\end{gathered}\end{equation*}
What makes this system interesting is the presence of the nonlinear function $g(x)$, which is defined as
\begin{equation}
\label{eq:nonlinear-Chua}
g(x) = G_b x + \frac{1}{2}(G_a - G_b)(|x + 1| - |x - 1|).
\end{equation}
This represents the three-segment piecewise-linear characteristic curve of Chua's diode, illustrated in~\Cref{fig:diode}. Furthermore, the set of parameters used by Galias to prove positive topological entropy is as follows.
\begin{itemize}
\item $C_1 = 1$; 
\item $C_2 = 9.3515$; 
\item $L = 0.06913$; 
\item $R = 0.33065$ (where $G = 1/R$); \item $G_a = -3.4429$;
\item $G_b = -2.1849$;
\item $R_0 = 0.00036$.
\end{itemize}

\Cref{fig:Chua-sim} depicts Chua's iconic double scroll attractor~\citep{Chua86:Double}.

\begin{figure}
    \centering
    \includegraphics[scale=0.55]{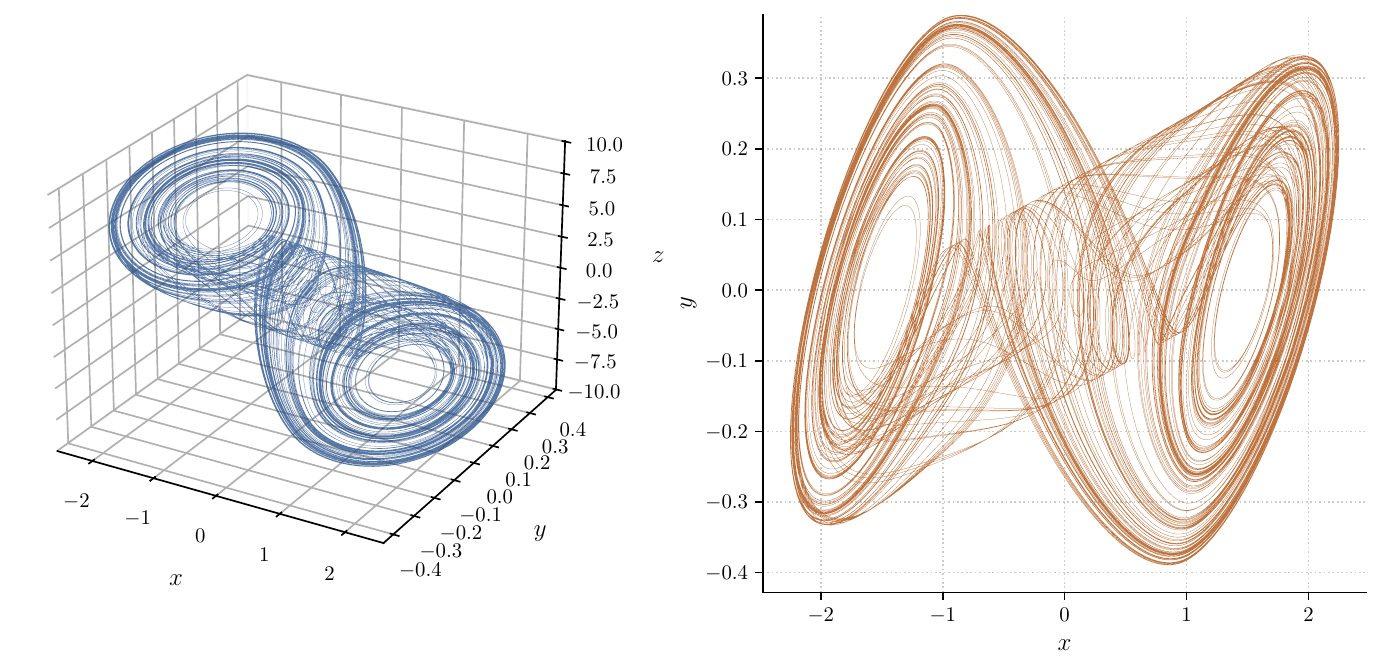}
    \caption{The double scroll attractor of Chua's circuit~\eqref{eq:Chua}.}
    \label{fig:Chua-sim}
\end{figure}

\paragraph{Physical interpretation} The above dynamical system can be obtained by applying Kirchhoff's laws to the electronic circuit depicted in~\Cref{fig:circuit} through the correspondence $x \leftrightarrow v_{C_1}, y \leftrightarrow v_{C_2}, z \leftrightarrow i_L$. The physical derivation is as follows.
\begin{itemize}
\item We first apply Kirchhoff's current law at the node joining $C_1$, $R$, and the Chua diode, which yields $i_G = i_{C_1} + i_g$. Since $i_{G} = G (v_{C_2} - v_{C_1})$ (current flows from the higher to the lower potential), $i_g = g(v_{C_1})$, and $i_{C_1} = C_1 \frac{d v_{C_1} }{dt}$, the first differential equation follows.
\item We then apply Kirchhoff's current law at the node joining $C_2$, $R$, and the inductor, which yields $i_{C_2} = - i_G + i_L$, and the second differential equation follows. 
\item Finally, we apply Kirchhoff's voltage law applied to the loop containing the inductor $L$ and the second capacitor $C_2$ to obtain the third equation.
\end{itemize}

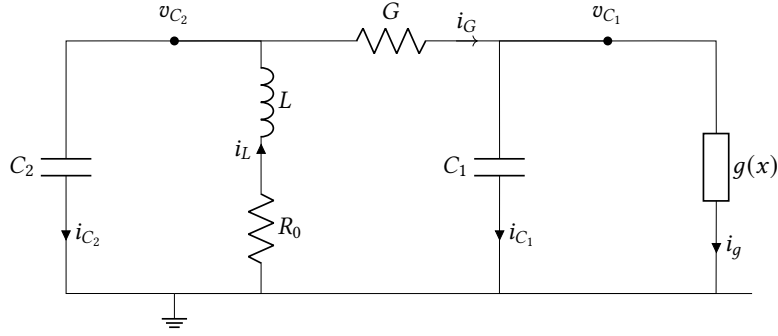
\begin{figure}
    \centering
    \scalebox{0.8}{\begin{circuitikz}[american, scale=1.2]

\tikzset{
  node/.style = {circle, fill, inner sep=1.5pt}
}

\coordinate (G) at (0,0);
\coordinate (Y) at (0,3.5);   
\coordinate (X) at (6,3.5);   

\draw (G) node[ground]{} -- (8,0);

\draw (Y) to[R, l=$G$] (X);

\draw[->]
  ($(Y)!0.65!(X)$) -- ($(Y)!0.70!(X)$)
  node[midway, above] {$i_G$};


\draw (Y) to[short] ++(-1.5,0)
          to[C, l_=$C_2$, i^=$i_{C_2}$] ++(0,-3.5)
          -- ++(1.5,0);

\draw (Y) to[short] ++(1.2,0)
          to[L, l=$L$, i_<=$i_L$] ++(0,-1.7)
          to[R, l=$R_0$] ++(0,-1.8)
          -- ++(-1.2,0);


\draw (X) to[short] ++(-1.5,0)
          to[C, l_=$C_1$, i^=$i_{C_1}$] ++(0,-3.5)
          -- ++(1.5,0);

\draw (X) to[short] ++(1.5,0)
          to[generic, l=$g(x)$, i^=$i_g$] ++(0,-3.5)
          -- ++(-1.5,0);

\draw (Y) node[node]{} node[above=4pt] {$v_{C_2}$};
\draw (X) node[node]{} node[above=4pt] {$v_{C_1}$};

\end{circuitikz}}
    \caption{Chua's electronic circuit.}
    \label{fig:circuit}
\end{figure}
\section{Omitted proofs}
\label{appendix:proofs}

This section contains the proofs omitted from the main body.

\subsection{Proofs from Section~\ref{sec:Chua}}
\label{sec:proofsChua}

We begin with the omitted proofs from~\Cref{sec:Chua}.

\nonlinear*

\begin{proof}
    We assume that there is a continuum of items, which we represent with a density function $\rho: [1.1, 1.9] \to \R_{\geq 0}$. We parameterize $\rho$ in terms of the reserve price attached to each item in the continuum, so that $\rho(p)$ is the density of items with reserve price $p$. Moreover, we assume that the autobidder values these items equally for 1 per unit. Let $m \in [1.1, 1.9]$ be the multiplier. The utility of the autobidder reads $u(m) = \int_{p = 1.1}^m ( 1 - p) \rho(p) dp$. As a result, if we set $\rho(p) = h'(p)/(1 - p) \geq 0$ for $p \in [1.1, 1.9]$, the claim follows.
\end{proof}

It is worth noting that the density constructed above may not be continuous as a function of the price. If one is willing to tolerate an arbitrarily small amount of error, which our main simulation result does, one can afford to use a continuous density (for example, by approximating each function $h_i$ using Weierstrass theorem).

\approxint*

\begin{proof}
Let $L > 0$ be a parameter such that $|m(t + \tau) - m(t)| \leq L \tau$ for any $t$ and $\tau$. We have
\begin{equation*}
    e^{-\lambda t} \int_{\tau=0}^t \lambda e^{\lambda \tau} ( 3 - m(\tau)) d \tau = e^{-\lambda t} \int_{\tau=0}^{t - \epsilon} \lambda e^{\lambda \tau} ( 3 - m(\tau)) d \tau + e^{-\lambda t} \int_{\tau=t - \epsilon}^t \lambda e^{\lambda \tau} ( 3 - m(\tau)) d \tau
\end{equation*}
for any $\epsilon \in (0, t)$. We bound the two resulting terms separately. First,
\begin{equation*}
    e^{-\lambda t} \int_{\tau=0}^{t - \epsilon} \lambda e^{\lambda \tau} ( 3 - m(\tau)) d \tau \leq 2 e^{-\lambda t} \int_{\tau=0}^{t-\epsilon} \lambda e^{\lambda \tau} d \tau \leq 2 e^{-\lambda \epsilon} - 2e^{-\lambda t}
\end{equation*}
since $m(\tau) \geq 1$. Second,
\begin{align*}
    e^{-\lambda t} \int_{\tau=t - \epsilon}^t \lambda e^{\lambda \tau} ( 3 - m(\tau)) d \tau &\leq e^{-\lambda t} \int_{\tau=t - \epsilon}^t \lambda e^{\lambda \tau} ( 3 - m(t)) d \tau + \epsilon L e^{-\lambda t} \int_{\tau = t - \epsilon}^t \lambda e^{\lambda \tau} \\
    &\leq (3 - m(t)) ( 1 - e^{- \lambda \epsilon}) + \epsilon L ( 1 - e^{- \lambda \epsilon}) \\
    &\leq 3 - m(t) - e^{-\lambda \epsilon} + \epsilon L
\end{align*}
since $m(t) \leq 2$. Combining and using the fact that $\widebar{m}(t) = e^{-\lambda t} \int_{\tau=0}^t \lambda e^{\lambda \tau} ( 3 - m(\tau)) d \tau + e^{- \lambda t} \widebar{m}(0)$,
\begin{align*}
    \widebar{m}(t) &\leq 3 - m(t) + \epsilon L + e^{-\lambda \epsilon} - 2 e^{-\lambda t} + e^{- \lambda t} \widebar{m}(0) \\
    &\leq 3 - m(t) + \epsilon L + e^{-\lambda \epsilon}.
\end{align*}
Similar reasoning bounds the difference from below as
\begin{equation*}
    \widebar{m}(t) \geq 3 - m(t) - \epsilon L - 2 e^{-\lambda \epsilon}.
\end{equation*}
As a result,
\begin{equation*}
    | \widebar{m}(t) -  (3 - m(t)) | \leq  \inf_{\epsilon \in (0, t) } \left( \epsilon L + 2 e^{- \lambda \epsilon} \right).
\end{equation*}
Now, let $\epsilon^* = \frac{1}{\lambda} \ln \left( \frac{2 \lambda }{L} \right)$ be the minimizer of the above. If $\epsilon^* \leq t$, it follows that
\begin{equation*}
    | \widebar{m}(t) -  (3 - m(t)) | \leq \Theta_\lambda \left( \frac{\ln \lambda}{\lambda} \right),
\end{equation*}
as desired. Otherwise, $t \leq \frac{1}{\lambda} \ln \left( \frac{2 \lambda }{L} \right)$. We have
\begin{align}
    \widebar{m}(t) &= e^{-\lambda t} \int_{\tau=0}^t \lambda e^{\lambda \tau} ( 3 - m(\tau)) d \tau + e^{- \lambda t} \widebar{m}(0) \notag \\
    &\leq (3 - m(t) + L t) e^{-\lambda t} \int_{\tau=0}^t \lambda e^{\lambda \tau} d \tau  + e^{- \lambda t} \widebar{m}(0) \notag \\
    &\leq (3 - m(t) + L t)(1 - e^{-\lambda t} ) + e^{- \lambda t} (3 - m(t) + L t) \notag \\
    &= 3 - m(t) + Lt.\label{align:final1}
\end{align}
Similarly,
\begin{align}
    \widebar{m}(t) &\geq (3 - m(t) - L t) ( 1 - e^{-\lambda t}) + e^{- \lambda t} \widebar{m}(0) \notag \\
    &\geq (3 - m(t) - L t) ( 1 - e^{-\lambda t}) + e^{- \lambda t} (3 - m(t) - L t) \notag \\
    &= 3 - m(t) - L t,\label{align:final2}
\end{align}
where we used the fact that $\widebar{m}(0) = 3 - m(0)$. Combining~\eqref{align:final1} and~\eqref{align:final2}, the proof follows.
\end{proof}

\simgen*

\begin{proof}
Starting from the original dynamical system, we consider the system
\begin{equation*}
\begin{gathered}
\frac{dx_i}{dt} = \langle \mat{A}_{i, :}, \vec{x} \rangle + h_i(x_i), \\
\frac{d \widebar{x}_i }{dt} = 3 \lambda - \lambda x_i(t) - \lambda \widebar{x}_i(t).
\end{gathered}
\quad i = 1, 2, \dots, d.
\end{equation*}
This can be equivalently cast as
\begin{equation}
    \label{eq:nonlinear-new}
\begin{gathered}
\frac{dx_i}{dt} = \sum_{j = 1}^n \left( \mathbbm{1} \{ \mat{A}_{i, j} > 0 \} \mat{A}_{i, j} x_j + \mathbbm{1} \{ \mat{A}_{i, j} < 0 \} \mat{A}_{i, j} x_j  \right)  + h_i(x_i), \\
\frac{d \widebar{x}_i }{dt} = 3 \lambda - \lambda x_i(t) - \lambda \widebar{x}_i(t).
\end{gathered}
\quad i = 1, 2, \dots, d.
\end{equation}
We now consider an approximation to~\eqref{eq:nonlinear-new} using~\Cref{lemma:approxint}:
\begin{equation}
    \label{eq:sub}
\begin{gathered}
\frac{dx_i}{dt} = \sum_{j = 1}^n \left( \mathbbm{1} \{ \mat{A}_{i, j} > 0 \} \mat{A}_{i, j} (3 - \widebar{x}_j) - \mathbbm{1} \{ \mat{A}_{i, j} < 0 \} (-\mat{A}_{i, j}) x_j  \right)  + h_i(x_i), \\
\frac{d \widebar{x}_i }{dt} = 3 \lambda - \lambda x_i(t) - \lambda \widebar{x}_i(t).
\end{gathered}
\quad i = 1, 2, \dots, d.
\end{equation}
What we did was to replace the non-competitive appearances of $x_j$ with $3 - \widebar{x}_j$ so that the linear component becomes competitive. This substitution is justified by~\Cref{lemma:approxint}. In particular, \Cref{assumption:nonlinear} together with~\Cref{lemma:approxint} imply that, for sufficiently large $\lambda$, $x_i(t) \in [1.05, 1.95]$ and $\widebar{x}_i(t) \in [1.05, 1.95]$. For convenience in the notation, we define, for each $i \in [d]$, $P_i \defeq \{ j \in [d] : \mat{A}_{i, j} > 0 \} $ and $N_i \defeq \{ j \in [d] : \mat{A}_{i, j} < 0 \} $. With this notation, we write~\eqref{eq:sub} as
\begin{equation*}
\begin{gathered}
\frac{dx_i}{dt} = \sum_{j \in P_i} - \mat{A}_{i, j} \widebar{x}_j - \sum_{j \in N_i} (-\mat{A}_{i, j} x_j) + 3 \sum_{j \in P_i} \mat{A}_{i, j}  + h_i(x_i), \\
\frac{d \widebar{x}_i }{dt} = 3 \lambda - \lambda x_i(t) - \lambda \widebar{x}_i(t).
\end{gathered}
\quad i = 1, 2, \dots, d.
\end{equation*}
We write this equivalently as
\begin{equation*}
\begin{gathered}
\frac{dx_i}{dt} = \sum_{j \in P_i} (2 \mat{A}_{i, j} - \mat{A}_{i, j} \widebar{x}_j) + \sum_{j \in N_i} (2\widebar{\mat{A}}_{i, j} - \widebar{\mat{A}}_{i, j} x_j) + \sum_{j \in P_i} \mat{A}_{i, j} + 2 \sum_{j \in N_i} \mat{A}_{i, j}  + h_i(x_i), \\
\frac{d \widebar{x}_i }{dt} = 3 \lambda - \lambda x_i(t) - \lambda \widebar{x}_i(t).
\end{gathered}
\quad i = 1, 2, \dots, d.
\end{equation*}
To make the purpose of the above step clear, we used the notation $\widebar{\mat{A}} = - \mat{A}$. Moreover, for each $i$ such that $i \in N_i$, we incorporate that term into the nonlinear function $h_i$; since the added linear function has negative derivative, this does not affect the application of~\Cref{lemma:nonlinear-sim}.

To complete the proof, we use certain ideas from the reduction of~\citet{Leme24:Complex}. For each variable $i \in [d]$ and $j \in P_i$, the term $2 \mat{A}_{i, j} - \mat{A}_{i, j} \widebar{x}_j$ can be simulated by having autobidder $i$ compete with autobidder $\widebar{j}$ for an item valued $2 \mat{A}_{i, j}$ and $\mat{A}_{i, j}$, respectively. Because of the invariance $x_i(t) \in [1.05, 1.95]$ and $\widebar{x}_i(t) \in [1.05, 1.95]$, it is always the former autobidder who procures the item, but the price is dictated by the other bidder, matching the desired linear term. Similarly, for each variable $i \in [d]$ and $j \in N_i$, the term $2 \widebar{\mat{A}}_{i, j} - \widebar{\mat{A}}_{i, j} x_j$ can be simulated by having the autobidder $i$ compete with autobidder $j$ for an item valued $2 \widebar{\mat{A}}_{i, j}$ and $\widebar{\mat{A}}_{i, j}$, respectively. (By construction, $\widebar{\mat{A}}_{i, j} > 0$ for all $j \in N_i$.) 

What remains is to treat the constant terms. A positive constant $C_i$ can always be incorporated into the utility of each autobidder $i$ by having them obtain an item of value $C_i$ without any competition. For a negative constant $C_i < 0$, we introduce an auxiliary autobidder $i'$ whose multiplier is initialized and remains at $2$. This can be achieved, for example, by having $i'$ always win an item of value $1$ with no competition and also win an item of value $1$ with reserve price $2$. Now, we can subtract a certain constant from $i$'s utility by having $i'$ value an item at $1$ while $i$ values it at $1.99$. It is always $i$ that wins this item, adding $-0.01$ in its utility. We can repeat this process with new auxiliary autobidders until the remaining constant $C_i'$ lies in $(-0.01, 0)$. Finally, we employ the same construction but $i$ now values the item for $2 + C_i'$.
\end{proof}

\subsection{Proofs from Section~\ref{sec:discrete}}
\label{sec:proofs-discrete}

We next establish~\Cref{lemma:period3}. To do so, we examine $F^{(1)}, F^{(2)},$ and $F^{(3)}$; see~\Cref{fig:compositions}.

\periodthree*

\begin{proof}
    We first claim that $F$ has a unique fixed point at $m = v$. That $v$ is a fixed point follows because $u(v) = 0$; when both autobidders pick a multiplier equal to $v$, the first autobidder wins the first item but pays $v$ for that item and the second autobidder wins the second item and pays $v$ for that item. For any multiplier $m > v$, the utility of each autobidder is negative since they pay more than $v$ for the item, whereas for any multiplier $m < v$ the utility of each autobidder is positive since each wins their favorite item and pays less than $v$. In particular, $F$ is given by $F(m) = m e^{\eta (v - m)}$. In other words, for any $m \neq v$ it holds that $F(m) \neq m$.

    We next turn to the analysis of $F^{(2)} = F(F(\cdot))$. Since $F(m) = m e^{\eta (v - m)}$, we have that $F(F(m)) = m e^{\eta (v - m)} e^{\eta(v - m e^{\eta (v - m)})}$. As a result, $F(m) = F(F(m)) \iff v - m e^{\eta (v - m)} = 0$. To analyze the solutions to this equation, we consider the function $g(m) \defeq v - m e^{\eta(v - m)}$. Its derivative reads $g'(m) = e^{\eta(v - m)} (\eta m - 1)$. So, since $g(0) > 0$ and $g(v) = 0$, it follows that $g$ has a root somewhere in the interval $(0, 1/\eta)$ and the obvious root at $m = v$. Thus, given that $\eta \geq 1$, $g$ has no root in $(1, v)$.

    We finally prove that $F^{(3)} = F(F(F(\cdot)))$ has a fixed point in $(1, v)$. For any $m > 1$, we let $m_1 = m e^{\eta (v - m)}$, $m_2 = m_1 e^{\eta u(m_1)} = m e^{\eta ( u(m_1) + u(m) )}$, and $m_3 = m e^{\eta (3v - (m_2 + m_1 + m) )}$. We analyze the solutions to the equation $3v = m + m e^{\eta (v - m) } + m e^{\eta (v - m) + \eta (v - m e^{\eta (v - m)}) }$. We define $g(m) \defeq 3v - m - m e^{\eta (v - m) } - m e^{\eta (v - m) + \eta (v - m e^{\eta (v - m)}) }$. It is easy to derive that $g'(m=v) = - 3 + 3 \eta v - \eta^2 v^2 < 0$, no matter the value of $\eta$. Furthermore, $g(1) < 3v - 1 - e^{\eta(v - 1)}$ since $e^{(\cdot)} > 0$. When $\eta \geq \frac{1}{v-1} \ln(3v - 1)$, it follows that $g(1) < 0$. Given that $g(v) = 0$ and $g'(m = v) < 0$, it follows that $g$ admits a root in $(1, v)$.

    Taken together, we have shown that there is an orbit with period $3$, as claimed.
\end{proof}

\Schwarzian*

\begin{proof}
    We first have
    \begin{equation*}
        F'(m) = \frac{d}{dm} \left[ m e^{\eta(v - m)} \right] = e^{\eta(v - m)} (1 - \eta m).
    \end{equation*}

The second derivative reads
\begin{equation*}
    F''(m) = \frac{d}{dm} \left[ e^{\eta(v - m)} (1 - \eta m) \right] 
= -\eta e^{\eta(v - m)} (1 - \eta m) + e^{\eta(v - m)} (-\eta) = e^{\eta(v - m)} (\eta^2 m - 2\eta).
\end{equation*}
Finally, we calculate the third derivative as
\begin{equation*}
    F'''(m) = \frac{d}{dm} \left[ e^{\eta(v - m)} (\eta^2 m - 2\eta) \right] = -\eta e^{\eta(v - m)} (\eta^2 m - 2\eta) + e^{\eta(v - m)} (\eta^2)
= e^{\eta(v - m)} (3\eta^2 - \eta^3 m).
\end{equation*}
We now combine those expressions to evaluate $S F (m)$ as follows.
\begin{align*}
    S F (m) = \frac{\eta^2(3 - \eta m)}{1 - \eta m} - \frac{3}{2} \left( \frac{-\eta(2 - \eta m)}{1 - \eta m} \right)^2 &= \frac{\eta^2(3 - \eta m)}{1 - \eta m} - \frac{3\eta^2(2 - \eta m)^2}{2(1 - \eta m)^2} \\
    &= - \frac{\eta^2}{2(1 - \eta m)^2} ( 6 - 4\eta m + \eta^2 m^2 ).
\end{align*}
    Since $6 - 4\eta m + \eta^2 m^2 = (\eta m - 2)^2 + 2 > 0$, it follows that $S F(m) < 0$ for any $m \neq 1/\eta$, and this is so no matter the choice of the learning rate $\eta$.
\end{proof}

\begin{figure}
    \centering
    \includegraphics[scale=0.6]{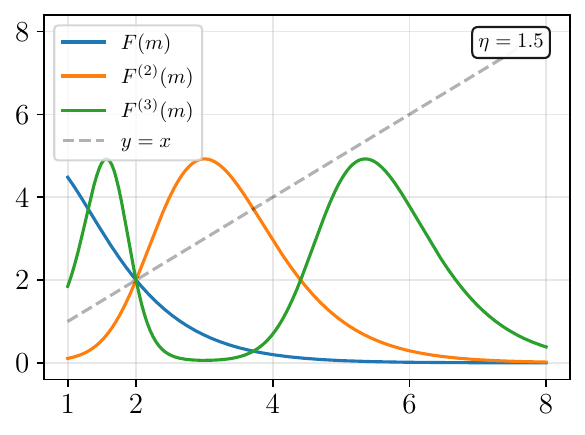}
    \includegraphics[scale=0.6]{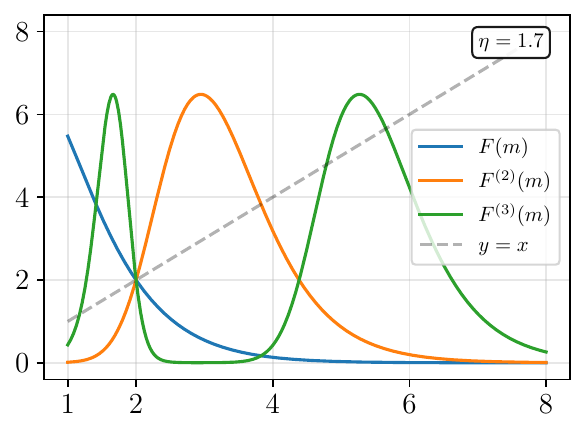}
    \caption{Compositions of $F = m e^{\eta (v - m)}$ for $\eta = 1.5$ (left) and $\eta = 1.7$ (right). $v$ is set to $2$. According to~\Cref{lemma:period3}, when $\eta \geq \frac{1}{v - 1} \ln (3v - 1) = \ln 5 \approx 1.6094$, $F^{(3)}$ has a fixed point in $(1, 2)$; this is confirmed in the plots.}
    \label{fig:compositions}
\end{figure}

\end{document}